\documentclass[11pt,a4paper]{article}

\usepackage{amsmath}
\usepackage{amssymb}
\usepackage{amsthm}
\usepackage{authblk}
\usepackage{bussproofs}
\usepackage{cancel}
\usepackage[T1]{fontenc}
\usepackage{mathpartir}
\usepackage{styles/tr-coqvf}
\usepackage{styles/tr-listings}
\usepackage{tabularx}
\usepackage{tikz}
\usepackage{xcolor}

\usetikzlibrary{trees}

\usepackage{hyperref}

\theoremstyle{plain}
\newtheorem{theorem}{Theorem}[section]
\newtheorem{lemma}[theorem]{Lemma}
\theoremstyle{definition}
\newtheorem{definition}[theorem]{Definition}
\newtheorem{notation}[theorem]{Notation}

\newcommand{\CL}{\ensuremath{\textrm{CL}}}

\newcommand{\vfcx}{VeriFast Cx}
\newcommand{\cbsem}{\texttt{cbsem}}

\newcommand{\coderef}[1]{\texttt{\detokenize{#1}}}

\author[1]{Stefan Wils}
\author[2]{Bart Jacobs}
\affil[1,2]{IMEC-DistriNet, Computer Science Dept., KU Leuven\\Celestijnenlaan 200A, 3001 Leuven, Belgium}
\affil[1]{{\footnotesize{stefan.wils@kuleuven.be}}}

\title{Certifying C program correctness with respect to CompCert with VeriFast}

\newcommand{\vfurl}{\url{https://doi.org/10.5281/zenodo.5585276}}

\begin{document}

\maketitle

\begin{abstract}
  VeriFast is a powerful tool for verification of various correctness properties
  of C programs using symbolic execution. However, VeriFast itself has not been
  verified. We present a proof-of-concept extension which generates a
  correctness certificate for each successful verification run individually.
  This certificate takes the form of a Coq script containing two proofs
  which, when successfully checked by Coq, together remove the need for trusting
  in the correctness of VeriFast itself.

  The first proves a lemma expressing the correctness of the program with
  respect to a big step operational semantics developed by ourselves, intended
  to reflect VeriFast's interpretation of C. We have formalized this semantics
  in Coq as \cbsem. This lemma is proven by symbolic execution in Coq, which in
  turn is implemented by transforming the exported AST of the program into a Coq
  proposition representing the symbolic execution performed by VeriFast itself.

  The second proves the correctness of the same C program with respect to
  CompCert's Clight big step semantics. This proof simply applies our proof of
  the soundness of \cbsem\ with respect to CompCert Clight to the first proof.
\end{abstract}

\tableofcontents

\section{Introduction}

VeriFast is a general verification tool for programs written in C
\cite{jacobs-vf-tutorial-2017}. If VeriFast reports ``0 errors found'', this is
intended to imply that no execution of the C program has undefined behavior,
such as accesses of unallocated memory or data races.

VeriFast performs its verification in a modular fashion by symbolically
executing each function individually. Whole program correctness is checked by
verifying that each function call satisfies the preconditions for that function,
which are provided by user annotations. VeriFast is written in OCaml and relies
on an SMT solver for checking the assertions which allow it to conclude that an
input program is correct with regards to its specification.

An open problem of course is that desirable qualities in a verification tool
such as speed, ease of use and realism in supported language features increase
the complexity of the verifier and grow the size of the code base which needs to
be trusted. This makes verification of VeriFast itself challenging and prevents
us from being sure that it is sound. So the question ultimately is:
\emph{what can the user really conclude when VeriFast reports "0 errors found"?}

This simple question involves more than ``merely'' straightforward
implementation bugs within VeriFast. When VeriFast concludes that a program
meets its specification, it has also made a large number of assumptions about
the semantics of C which are not explicitly provided by the user's annotations.

In effect, the sum of these assumptions results in a version of C specific to
VeriFast. We need to trust that this ``VeriFast C'' relates meaningfully to the
version of C implemented by a particular compiler. In other words, VeriFast C
should accept a usable but rather strict subset of programs that have defined
behavior according to the ISO C standard and especially of the subset of
programs which are accepted by this compiler.

More broadly, VeriFast relies on an implicit metatheory which must be assumed
correct when trusting its conclusions. To provide a simple example of such a
``metalemma'': VeriFast concludes that the result of a division of two
\vfcinline{int}s $n$ and $d$ is itself a valid \vfcinline{int} and not
out-of-bounds ($\vfMinSigned \leq n / d \leq \vfMaxSigned$) if: $n$ and $d$
themselves are not out-of-bounds; $d \neq \mathtt{0}$; and $\neg (n =
\vfMinSigned\ \land\ d = \mathtt{-1})$. From these three conditions it follows
that the division result is not out-of-bounds. But the actual proof for this
fact is never explicitly checked within VeriFast.

The question then becomes: can we provide an explicit list of all the
assumptions that were made in implementing VeriFast? Better yet: can we prove
\emph{all of them} correct, preferably in a formal proof checker such as Coq?
Finally, can we actually prove them correct
\emph{with respect to a third party formal semantics of C?}

In earlier work, Vogels \emph{et al.} \cite{vogels-fwvf-lmcs-2015} already
provided a precise definition in Coq for a subset of VeriFast's symbolic
execution and proved that subset sound with respect to concrete program
execution.

In this technical report, we present a different approach by providing a
certification \emph{for each individual run of VeriFast}. We extend VeriFast
such that, for each successful verification, it \emph{exports} a machine
checkable proof for the correctness of the program, not only with regards to our
own formal semantics (called \cbsem), but with regards to CompCert's
\emph{Clight semantics} as well.

CompCert \cite{leroy-compcert-cacm-2009} is a formally verified C compiler which
is guaranteed to not introduce new behaviors into a program, ensuring that the
guarantees provided by a formal verification tool will be preserved during
compilation. Clight \cite{blazy-clight-jar-2009} is an intermediate C-like
language used within CompCert that has pure expressions and implements
assignments and function calls as statements. Therefore, proving correctness of
the program with regards to Clight allows a VeriFast user to \emph{entirely}
shift trust from VeriFast to Coq's computational kernel and to CompCert Clight
semantics.

\subsection{Overview}

The work that is presented in this report provides preliminary validation of our
approach for a limited subset of VeriFast and C language features.

We begin our report in Section \ref{section:example} with a concrete
demonstration of our approach for a simple C program. We explain how VeriFast
symbolically executes that program and explore the contents of the Coq artefact
exported by VeriFast after verification of the program.

Section \ref{section:vfcx} begins by delimiting the subset of C language
features and VeriFast constructs that we chose for our proof of concept. This
subset is presented as an intermediate language, called \vfcx, which is
substantially similar to C but with some small tweaks to facilitate our work in
Coq. The chosen C subset is strong enough to encode terminating and diverging
executions and to allow expression of some simple undefined behavior,
specifically integer overflows and division by zero.

Section \ref{section:symexec} continues with presenting our first main result,
which consists of a Coq function $\vfSymExecFunc$ that takes a \vfcx\ function
together with its specification and generates a shallow embedding of VeriFast's
symbolic execution for that function into a Coq proposition. This proposition is
called the \emph{symbolic execution proposition} or SEP. We then describe our
current approach towards automatically proving the SEP in Coq. Proving the SEP
remains easy for now, exactly because of the limited scope of \vfcx.

Section \ref{section:cbsem} presents \cbsem, which is our formal operational
semantics for \vfcx\ expressions, statements and functions. We chose big step
semantics to get started quickly. \cbsem\ is actually comprised of two sets of
rules, one for terminating programs and one for diverging programs. We present a
singular predicate $\vfCbsemFunc$ describing \cbsem's notion of correctness of a
function. We conclude Section \ref{section:cbsem} with a theorem proving the
soundness of $\vfSymExecFunc$ with regards to $\vfCbsemFunc$, providing us with
a method to prove correctness with regards to \cbsem\ by symbolic execution.

Section \ref{section:clight} sets up a relation between \cbsem\ and Clight's big
step semantics for stores, expressions, statements and functions. These
relations are set up for a slightly smaller subset of \vfcx\ which is still big
enough to include the example program from Section \ref{section:example}. We
then provide custom-made notions of \emph{program correctness} for both \cbsem\
and Clight, and conclude with a soundness theorem stating that correctness of a
program in Clight follows from correctness of the equivalent program in \cbsem\
(and therefore by symbolic execution in VeriFast).

We end this technical report with a brief discussion of related work (Section
\ref{section:relatedwork}) and future work (Section \ref{section:futurework}).
Appendix~\ref{appendix:building_and_using} provides instructions on building and
using our development. Appendix~\ref{appendix:overview} gives an overview of the
modules found in our Coq development and links the concise notation used in this
report to the names used in the Coq code. Finally,
Appendix~\ref{appendix:listings} shows the entire Coq script that is exported by
our extension for the example from Section~\ref{section:example}.

\subsection{Source code}

The source code for our VeriFast extension, effectively a fork of VeriFast, can
be downloaded here:

\begin{center}
  \vfurl
\end{center}

Our extension generates Coq code and also includes a small Coq library. In this
technical report, we will use two methods for rendering such Coq code:
\begin{enumerate}
\item Larger chunks of Coq code will be rendered verbatim, similar to how we
will render C code.
\item For all other cases, in the interest of readability, we will use a more
concise form of notation. \vfcx\ ASTs embedded within these concisely rendered
Coq terms will likewise use a short hand form, marked by use of a different
font.
\end{enumerate}
Both forms of notation will be first introduced in Section
\ref{section:example}.


\section{An example}
\label{section:example}

In this section we illustrate the ideas of our report by applying them to a
concrete program. We begin by presenting the program and illustrating its
symbolic execution in VeriFast. We then use this example program to showcase the
actual workflow of our approach to certified verification.

\subsection{Verifying a program with symbolic execution}
\label{subsection:example_program}

Listing \ref{lst:test_countdown.c} shows \coderef{tests/coq/test_countdown.c}, a
C program implementing a simple countdown loop which is included with our
VeriFast extension.

\vfclistinglabel{sources/test_countdown.c}{\coderef{tests/coq/test_countdown.c}}{lst:test_countdown.c}

When compiling this program and executing it, the memory location denoted by
\vfcinline{x} gets value 32767. The body of the \vfcinline{while} loop will be
executed 32767 times, at which point the loop guard detects that \vfcinline{x =
0}, the loop is exited and the program terminates with a return value of 0.

In contrast to this familiar \emph{concrete} mode of execution, VeriFast
performs \emph{symbolic} execution. This means that a variable such as
\vfcinline{x} does not always have a concrete value, such as 32767. Instead, it
may receive a \emph{symbolic} value $\varsigma \in \mathit{Symbols}$, the
possible interpretations of which are \emph{constrained} by a \emph{path
condition}. This path condition comprises a set of assumptions that further
narrows down the possible interpretations for $\varsigma$, based e.g.~on the
preconditions (specified at the beginning of the function by the
\vfcinline{requires} keyword), branching conditions (such as the loop condition
\vfcinline{x < 0}) and user-provided loop invariants (specified using the
\vfcinline{invariant} keyword).

\newcommand{\symexecbullet}{\ensuremath{\color{sol-orange} \bullet}}

\newcommand{\ASM}[1]{\ensuremath{\{ #1 \}}}
\newcommand{\ASMnew}[1]{\ensuremath{{\color{sol-green}#1}}}

\newcommand{\STR}[1]{\ensuremath{\langle[ #1 ]\rangle}}
\newcommand{\STRnew}[1]{\ensuremath{{\color{sol-green}#1}}}

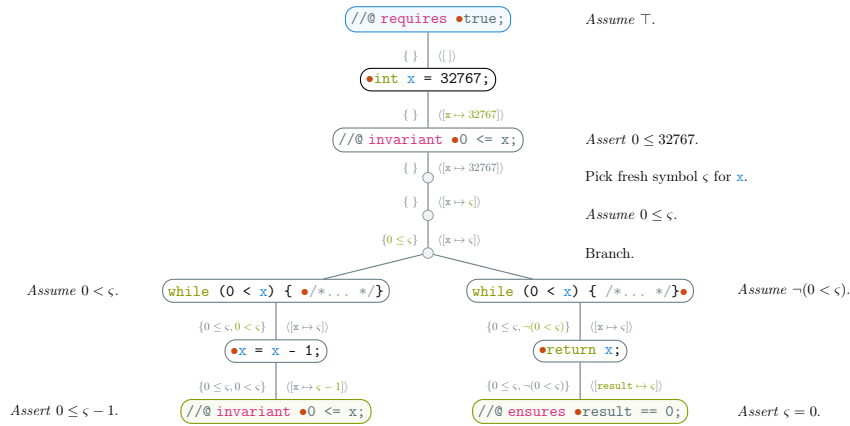
\begin{figure}[!ht]
  \centering
  \tikzstyle{symexec-node} = [draw]
  \tikzstyle{symexec-node-start} = [draw, color=sol-blue, fill=sol-blue!5]
  \tikzstyle{symexec-node-done} = [draw, color=sol-green, fill=sol-green!5]
  \tikzstyle{symexec-node-circ} = [circle, scale=0.7, draw, color=sol-base00!70, fill=sol-base00!10]
  \tikzstyle{symexec-edge} = [color=sol-base00]
  \tikzstyle{symexec-edge-l} = [left, scale=0.7, xshift=-0.2cm, yshift=-0.2cm, text=sol-base00]
  \tikzstyle{symexec-edge-r} = [right, scale=0.7, xshift=0.2cm, yshift=-0.2cm, text=sol-base00]
  \begin{tikzpicture}[
      every node/.style={text centered, rounded corners, scale=0.55},
      level 3/.style={level distance=0.5cm},
      level 7/.style={level distance=0.8cm},
      level distance=0.8cm,
      sibling distance=4cm
    ]
    \node (Root) [symexec-node-start] {\vfcinline{//@ requires $\ \symexecbullet$ true;}}
    child {
      node [symexec-node] {\vfcinline{$\symexecbullet$ int x = 32767;}}
      child {
        node [symexec-node] {\vfcinline{//@ invariant $\ \symexecbullet$ 0 <= x;}}
        child {
          node[symexec-node-circ] {}
          child {
            node[symexec-node-circ] {}
            child {
              node[symexec-node-circ] {}
              child {
                node [symexec-node] {\vfcinline{while (0 < x) \{ $\ \symexecbullet$ /* ... */ \}}}
                child {
                  node [symexec-node] {\vfcinline{$\symexecbullet$ x = x - 1;}}
                  child {
                    node  (LastL) [symexec-node-done] {\vfcinline{//@ invariant $\ \symexecbullet$ 0 <= x;}}
                    edge from parent [symexec-edge] {
                      node[symexec-edge-l] {$\ASM{ 0 \leq \varsigma, 0 < \varsigma }$}
                      node[symexec-edge-r] {$\STR{ \mathtt{x} \mapsto \STRnew{\varsigma - 1} }$}
                    }
                  }
                  edge from parent [symexec-edge]  {
                    node[symexec-edge-l] {$\ASM{ 0 \leq \varsigma, \ASMnew{0 < \varsigma} }$}
                    node[symexec-edge-r] {$\STR{ \mathtt{x} \mapsto \varsigma }$}
                  }
                }
                edge from parent [symexec-edge]
              }
              child {
                node [symexec-node] {\vfcinline{while (0 < x) \{ /* ... */ \} $\symexecbullet$}}
                child {
                  node [symexec-node] {\vfcinline{$\symexecbullet$ return x;}}
                  child {
                    node (LastR) [symexec-node-done] {\vfcinline{//@ ensures $\ \symexecbullet$ result == 0;}}
                    edge from parent [symexec-edge] {
                      node[symexec-edge-l] {$\ASM{ 0 \leq \varsigma, \neg (0 < \varsigma) }$}
                      node[symexec-edge-r] {$\STR{ \STRnew{\mathtt{result} \mapsto \varsigma} }$}
                    }
                  }
                  edge from parent [symexec-edge] {
                    node[symexec-edge-l] {$\ASM{ 0 \leq \varsigma, \ASMnew{\neg (0 < \varsigma)} }$}
                    node[symexec-edge-r] {$\STR{ \mathtt{x} \mapsto \varsigma }$}
                  }
                }
                edge from parent [symexec-edge]
              }
              edge from parent [symexec-edge] {
                node[symexec-edge-l] {$\ASM{ \ASMnew{0 \leq \varsigma} }$}
                node[symexec-edge-r] {$\STR{ \mathtt{x} \mapsto \varsigma }$}
              }
            }
            edge from parent [symexec-edge] {
              node[symexec-edge-l] {$\ASM{ \ }$}
              node[symexec-edge-r] {$\STR{ \mathtt{x} \mapsto \STRnew{\varsigma} }$}
            }
          }
          edge from parent [symexec-edge] {
            node[symexec-edge-l] {$\ASM{ \ }$}
            node[symexec-edge-r] {$\STR{ \mathtt{x} \mapsto 32767 }$}
          }
        }
        edge from parent [symexec-edge] {
          node[symexec-edge-l] {$\ASM{ \ }$}
          node[symexec-edge-r] {$\STR{ \STRnew{\mathtt{x} \mapsto 32767} }$}
        }
      }
      edge from parent [symexec-edge] {
        node[symexec-edge-l] {$\ASM{ \ }$}
        node[symexec-edge-r] {$\STR{ \ }$}
      }
    };

    \begin{scope}[
      every node/.style={right,scale=0.5}
    ]
      \path (Root           -| LastR) node {\emph{Assume} $\top$.};
      \path (Root-1-1       -| LastR) node {\emph{Assert} $0 \leq 32767$.};
      \path (Root-1-1-1     -| LastR) node {Pick fresh symbol $\varsigma$ for \vfcinline{x}.};
      \path (Root-1-1-1-1   -| LastR) node {\emph{Assume} $0 \leq \varsigma$.};
      \path (Root-1-1-1-1-1 -| LastR) node {Branch.};
    \end{scope}

    \begin{scope}[
      every node/.style={left,scale=0.5,xshift=-4cm}
    ]
      \path (Root-1-1-1-1-1-1     -| LastL) node {\emph{Assume} $0 < \varsigma$.};
      \path (LastL -| LastL) node {\emph{Assert} $0 \leq \varsigma - 1$.};
    \end{scope}

    \begin{scope}[
      every node/.style={right,scale=0.5,xshift=4cm}
    ]
      \path (Root-1-1-1-1-1-2     -| LastR) node {\emph{Assume} $\neg(0 < \varsigma)$.};
      \path (LastR -| LastR) node {\emph{Assert} $\varsigma = 0$.};
    \end{scope}
  \end{tikzpicture}
  \caption{A structural overview of symbolic execution in VeriFast.}
  \label{fig:symexec}
\end{figure}

Figure \ref{fig:symexec} shows in detail how VeriFast symbolically executes the
above program. Each \emph{node} of this tree can be considered a breakpoint in
the symbolic execution performed by VeriFast (marked by the red dot), associated
with a \emph{state}. This state is shown above the nodes, with the current path
condition on the left and the current \emph{symbolic store} on the right. The
symbolic store links each C variable to its current concrete or symbolic value.
Each edge represents a state transition which can change the path condition and
symbolic store; these changes are rendered in green.

Since \vfcinline{main} is given the weakest possible precondition
(\vfcinline{true} or $\top$), initially nothing can be assumed. In the next
line, the variable \vfcinline{x} is initialized with a literal value, so the
store now records a concrete value for \vfcinline{x}. Things get more
interesting in the \vfcinline{while} loop: first VeriFast will assert that the
loop invariant holds, which in this case is easy to see. After this assertion,
the value for \vfcinline{x} is replaced by a fresh symbolic value $\varsigma \in
\mathit{Symbols}$, for which the loop invariant \vfcinline{0 <= x} is assumed,
meaning that we know $0 \leq \varsigma$, whatever the value of $\varsigma$ is
otherwise. The symbolic execution now proceeds by splitting in two separate
execution branches:
\begin{itemize}
  \item The left branch executes \emph{all possible iterations of the loop} at
  once\footnote{In fact, depending on the size of type \vfcinline{int}, it
  executes many more iterations than would be possible in concrete execution. In
  effect, symbolic execution embodies an \emph{overapproximation} of all
  possible concrete executions.}: it tracks \emph{all} executions in which the
  loop condition \vfcinline{0 < x} holds. This means we assume $0 < \varsigma$.
  The expression statement \vfcinline{x = x - 1} leads to an update in the
  symbolic store. If \vfcinline{x} had symbolic value $\varsigma$ before, it now
  has value $\varsigma - 1$. Having reached the end of the loop body, VeriFast
  performs another assertion that the loop invariant \vfcinline{0 <= x} holds.
  Indeed, the SMT solver within VeriFast concludes that $0 \leq \varsigma - 1$,
  based on the assumptions about $\varsigma$ currently found in the path
  condition.

  \item The right branch executes \emph{the code after the loop}. This means
  that the loop condition does not hold, so our path condition now assumes $\neg
  (0 < \varsigma)$. It executes the \vfcinline{return x} statement, which leads
  to the symbolic store receiving a variable \vfcinline{result}, whose symbolic
  value is also $\varsigma$. Exiting the function means that we have to assert
  that the postcondition holds (specified by the \vfcinline{ensures} keyword).
  This leads to the SMT solver in VeriFast checking whether $\varsigma = 0$,
  which indeed it can conclude from the assumptions collected in the path
  condition.
\end{itemize}

Since both the loop body and the loop continuation satisfy their respective
postconditions, VeriFast concludes that the function (and the entire program) is
correct in terms of the C language and of its own specification.

\subsection{Certified program verification}

We will now demonstrate our approach to certified program verification using the
program shown in Listing \ref{lst:test_countdown.c}. (Appendix
\ref{appendix:building_and_using} provides further instructions on building and
using the extension on your own machine.) Our approach requires three steps:
\begin{enumerate}
  \item First, we run \texttt{clightgen} on the program. This tool, included
  with the CompCert distribution, generates a Coq script (\coderef{test_cc.v})
  containing the abstract syntax tree for the program in the Clight intermediate
  language:
  \begin{verbatim}$ clightgen tests/coq/test_countdown.c -o test_cc.v\end{verbatim}

  \item Second, we run VeriFast with the new \texttt{-emit\_coq\_proof} option
  on the program. This generates a second Coq script containing the
  VeriFast-generated correctness proof (\coderef{test_vf.v}):
  \begin{verbatim}$ bin/verifast -shared -emit_coq_proof -bindir bin
      tests/coq/test_countdown.c\end{verbatim}

  \item \coderef{test_vf.v} imports the \texttt{clightgen}-generated AST from
  \coderef{test_cc.v}. So finally, we can let Coq check our proof by compiling
  \emph{both} scripts:
  \begin{verbatim}$ coqc test_cc.v
$ coqc test_vf.v -Q src/coq verifast\end{verbatim}
\end{enumerate}
Let's inspect the contents of both scripts in more detail.





\subsection{The Coq script exported by \texttt{clightgen}}

When we run CompCert's \texttt{clightgen} tool on the above example, the
resulting Coq script \coderef{test_cc.v} contains, among other things, a record
\coqinline{f_main} describing the function body, arguments and return type of
function \vfcinline{main}. The function body \coqinline{f_main.fn_body} is a
Clight AST. For the example program, \texttt{clightgen} generates the following
AST in \coderef{test_cc.v}:
\begin{lstlisting}[
  language=Coq,
  style=coqfootnotenumbered,
  linerange={80-103},
  firstnumber=80
]
Definition f_main := {|
  fn_return := tint;
  fn_callconv := cc_default;
  fn_params := nil;
  fn_vars := nil;
  fn_temps := ((_x, tint) :: nil);
  fn_body :=
(Ssequence
  (Ssequence
    (Sset _x (Econst_int (Int.repr 32766) tint))
    (Ssequence
      (Swhile
        (Ebinop Olt (Econst_int (Int.repr 0) tint) (Etempvar _x tint) tint)
        (Sifthenelse (Ebinop Ole (Etempvar _x tint)
                       (Econst_int (Int.repr 5) tint) tint)
          (Sset _x
            (Ebinop Osub (Etempvar _x tint) (Econst_int (Int.repr 1) tint)
              tint))
          (Sset _x
            (Ebinop Osub (Etempvar _x tint) (Econst_int (Int.repr 5) tint)
              tint))))
      (Sreturn (Some (Etempvar _x tint)))))
  (Sreturn (Some (Econst_int (Int.repr 0) tint))))
|}.
\end{lstlisting}
As we will see in the next section, it is this AST that will be imported from
within the script exported by VeriFast. The rest of \coderef{test_cc.v} consists
mostly of declarations of external functions which do not interest us.

\subsection{The Coq script exported by VeriFast}

A full listing of the Coq script exported by VeriFast (\texttt{test\_vf.c}) can
be found in Listing \ref{lst:test_vf.v} in Appendix \ref{appendix:listings}.
Inspecting this script reveals the following major sections:
\begin{enumerate}
  \item \emph{An import of the Clight script}.

  \item \emph{The AST} of the program as parsed by VeriFast, corresponding to
  what we find in the \texttt{clightgen} export, but in the \vfcx\ language.
  Notably, it includes the VeriFast annotations, specifically the precondition
  and postcondition for \vfcinline{main} and the loop invariant.

  \item \emph{A lemma} stating that the program represented by this data
  structure is correct with regards to its own included specification and with
  regards to VeriFast's own operational semantics \cbsem. The proof for this
  lemma reduces the goal to proving success of our Coq encoding of VeriFast's
  symbolic execution. The rest of the proof consists of tactic applications
  discharging the proof obligations of the symbolic execution; it is generated
  by instrumenting VeriFast's symbolic execution engine.

  \item \emph{A proof goal} stating that the imported Clight program is correct
  with regards to Clight's \emph{own big step operational semantics}, proven by
  application of the above lemma.
\end{enumerate}
Let's consider each of these sections in more detail.

\subsubsection{Importing the Clight script}

First, we import the Clight script from \coderef{test_cc.v} and use it to
construct a ``local'' Clight program \coqinline{test_cc_prog}:
\begin{lstlisting}[
  language=Coq,
  style=coqfootnotenumbered,
  linerange={25-32},
  firstnumber=25
]
Require test_cc.

Definition test_cc_prog: Clight.program := mkprogram
  test_cc.composites
  [(test_cc._main, Gfun(Internal test_cc.f_main))]
  [test_cc._main]
  test_cc._main
  Logic.I.
\end{lstlisting}
This \coqinline{test_cc_prog} is like the program found in the Clight export,
but it does not contain the large set of external global definitions that Clight
adds.

\subsubsection{Program AST in \vfcx}
\label{subsubsection:program_ast_in_vfcx}

Next, we look at the description of the program as parsed by VeriFast. Since we
currently only support programs with a single main function taking no arguments,
the data structure consists \emph{only} of the AST of that function's body. The
AST for the example program in \coderef{test_vf.v} looks like this:

\begin{lstlisting}[language=Coq, style=coqfootnotenumbered, linerange={36-60},
  firstnumber=36]
Definition main_vf_func :=
  (StmtLet "x" ((ExprIntLit (32767))) (
    (
      (StmtWhile
        (ExprLt (ExprIntLit (0)) (ExprVar "x"))
        (ExprLe (ExprIntLit (0)) (ExprVar "x"))
        (
          (StmtBlock
            (
              (StmtExpr (ExprAssign
                (ExprVar "x")
                (ExprSub (ExprVar "x") (ExprIntLit (1)))));
              StmtSkip
            )
          );
          StmtSkip
        )
        StmtSkip
      );
      (
        (StmtReturn (ExprVar "x"));
        StmtSkip
      )
    )
  )).
\end{lstlisting} As mentioned in the introduction, in the
  rest of this text we will often use a more concise notation system for Coq
  code and, within that concise notation system, a special font for \vfcx\ ASTs
  embedded in Coq. This concise notation renders the above Coq definition as
  follows:
\begin{align*}
  \coqT{main_vf_func} :=\
  & \coqs{0} \vfxLet{\vfxVar{x}}{\vfxIntLit{32767}}{\vfxSeq{ \\
  & \coqs{0} \vfxWhile{(\vfxBinop{\vfxIntLit{0}}{<}{\vfxVar{x}})}{(\vfxBinop{\vfxIntLit{0}}{\leq}{\vfxVar{x}})}{ \\
  & \coqs{1}   \vfxSeq{\vfxBlock{\vfxSeq{\vfxBinop{\vfxVar{x}}{=}{\vfxBinop{\vfxVar{x}}{-}{\vfxIntLit{1}}}}{\vfxSkip}}}{\vfxSkip} \\
  & \coqs{0} }}{\vfxSeq{\vfxReturn{\vfxVar{x}}}{\vfxSkip}}}.
\end{align*}
\vfcx\ itself will be specified in full detail in Section \ref{section:vfcx}.

\subsubsection{Correctness with regards to \cbsem}

Below the definition of \coqinline{main_vf_func} we find the lemma stating
correctness of \coqinline{main_vf_func} in terms of \cbsem, the operational
semantics for VeriFast programs:
\begin{lstlisting}[
  language=Coq,
  style=coqfootnotenumbered,
  linerange={64-69},
  firstnumber=64
]
Lemma main_cbsem_by_symexec_tactics:
  cbsem.exec_func_correct
    (nil)
    ExprTrue
    main_vf_func
    (ExprEq (ExprVar "result") (ExprIntLit (0))).
\end{lstlisting}
Translating this into our more concise notation system, the lemma has the
following proof goal (see Definition \ref{def:cbsem_exec_func_correct} for
$\vfCbsemFunc$):
\begin{align*}
  \vdash \coqAAAA{\vfCbsemFunc}{\coqT{[}\ \coqT{]}}{\vfxTrue}{\coqT{main_vf_func}}{(\vfxBinop{\vfxVar{result}}{==}{\vfxIntLit{0}})}.
\end{align*}
This means that we have to provide, for each possible concrete execution of the
main function, starting from an initial state satisfying the weak precondition
\vfcinline{true}, a derivation showing that the execution either diverges or
terminates with a return value satisfying the postcondition.

Instead of constructing these derivations directly, the proof for this lemma
begins by applying the soundness result from Theorem
\ref{thm:symexec_func_sound} which states that correctness of a function in the
operational semantics \cbsem\ can be proven by symbolic execution of that
function:
\begin{lstlisting}[
  language=Coq,
  style=coqfootnotenumbered,
  linerange={64-64},
  firstnumber=64
]
Lemma main_cbsem_by_symexec_tactics:
\end{lstlisting}
$\vdots$
\begin{lstlisting}[
  language=Coq,
  style=coqfootnotenumbered,
  linerange={70-72},
  firstnumber=70
]
Proof.
  apply symexec_cbsem.symexec_func_sound.
  repeat (autounfold with vf_symexec_clight; simpl).
\end{lstlisting}
Application of this theorem gives us a new proof goal containing the symbolic
execution predicate $\vfSymExecFunc$ (see Definition \ref{def:sym_exec_func} for
$\vfSymExecFunc$):
\begin{align*}
  \vdash \coqAAAA{\vfSymExecFunc}{\coqT{[}\ \coqT{]}}{\vfxTrue}{\coqT{main_vf_func}}{(\vfxBinop{\vfxVar{result}}{==}{\vfxIntLit{0}})}.
\end{align*}

The goal now states that symbolic execution of the function must succeed. As
will be seen in \ref{subsection:sep}, this symbolic execution consists of
constructing a symbolic execution proposition or SEP. Reduction of the current
proof goal (performed by the repeated application of the \coqinline{autounfold}
and \coqinline{simpl} tactics in line 72) effectively computes this SEP for
\coqinline{main_vf_func} with the stated pre- and postcondition:
\begin{align*}
  \vdash \coqTrue \to\
  & \coqs{0} (0 \le 32767) \land (\\
  & \coqs{1}   \forall\ \coqT{sym}: \coqZ, (\vfMinSigned \le \coqT{sym} \le \vfMaxSigned) \to \\
  & \coqs{2}     (0 \le \coqT{sym}) \to ( \\
  & \coqs{3}       (0 < \coqT{sym}) \to \\
  & \coqs{4}         (\vfMinSigned \le \coqT{sym} - 1) \land (\coqT{sym} - 1 \le \vfMaxSigned)\ \land \\
  & \coqs{4}         (0 \le \coqT{sym} - 1) \land \coqTrue \\
  & \coqs{2}     ) \land ( \\
  & \coqs{3}       \neg (0 < \coqT{sym}) \to \\
  & \coqs{4}         \coqT{sym} = 0 \land \coqTrue \\
  & \coqs{2}     )).
\end{align*}
The structure of this Coq proof goal corresponds to the structure of the graph
in Figure \ref{fig:symexec}, with universal quantification taking the place of
symbol picking. The rest of the proof is then provided by a tactic script
exported by VeriFast, based on recordings of calls to the internal SMT solver:
\begin{lstlisting}[language=Coq, style=coqfootnotenumbered, linerange={64-64},
firstnumber=64]
Lemma main_cbsem_by_symexec_tactics:
\end{lstlisting} $\vdots$ 
\begin{lstlisting}[language=Coq,
style=coqfootnotenumbered, linerange={70-85}, firstnumber=70]
Proof.
  apply symexec_cbsem.symexec_func_sound.
  repeat (autounfold with vf_symexec_clight; simpl).
  simpl.
  intros.  (* _: True *)
  split.
  (* assert 0 <= 32767 *)
  lia.
  intros.  (* x: Z *)
  intros.  (* _: 0 <= x *)
  split.
  intros.  (* _: 0 < x *)
  split.
  (* assert INT_MIN <= (x - 1) *)
  lia.
  split.
\end{lstlisting}
$\vdots$ 
\begin{lstlisting}[language=Coq, style=coqfootnotenumbered,
linerange={99-99}, firstnumber=99]
Qed.
\end{lstlisting} We do not show the entire
proof for the sake of brevity (again, see Appendix \ref{appendix:listings} for
that), but it is clear that we only use three tactics: \coqinline{intros},
\coqinline{split} and \coqinline{lia}. Subsection \ref{subsection:proving_sep}
provides more details on the relation between these tactics and the SEP
structure.

\subsubsection{Correctness with regards to Clight}

The fourth and final part of the Coq script consists of a direct proof goal
stating that the program \coqinline{test_cc_prog} is correct, using a predicate
for Clight-based program correctness introduced in Definition
\ref{def:compcert_exec_prog_correct}. We show at once the proof goal and its
actual proof:
\begin{lstlisting}[
  language=Coq,
  style=coqfootnotenumbered,
  linerange={103-117},
  firstnumber=103
]
Goal cbsem_clight.compcert_exec_prog_correct test_cc_prog.
Proof.
  apply (cbsem_clight.vf_cl_sound main_vf_func).
  - repeat econstructor.
    unfold Globalenvs.Genv.init_mem. simpl.
    case_eq (Memory.Mem.alloc Memory.Mem.empty 0 1). intros.
    destruct (Memory.Mem.range_perm_drop_2 m b 0 1 Memtype.Nonempty).
    + unfold Memory.Mem.range_perm.
      intros.
      apply Memory.Mem.perm_alloc_2 with (1:=H) (2:=H0).
    + rewrite e.
      reflexivity.
  - eapply cbsem_clight.cbsem_func_sound.
    eapply main_cbsem_by_symexec_tactics.
Qed.
\end{lstlisting}
In essence, this expresses that the program either terminates with some return
value or diverges according to Clight's own big step semantics.

The proof for this statement is boilerplate, in the sense that it can be
automatically constructed \emph{by the same set of Coq tactics for every
supported program}. It commences by application of theorem
\coqinline{vf_cl_sound} in line 105, corresponding to our second soundness
theorem (Theorem \ref{thm:vf_cl_sound}). As a result, we now need to prove two
subgoals:
\begin{enumerate}
  \item The existence of a correspondence relation between
  \coqinline{test_cc_prog} and a slightly transformed version of
  \coqinline{main_vf_func} (see Definition \ref{def:prog_equiv} for more
  details). The proof for this subgoal involves proving that the Clight and
  VeriFast ASTs ``line up'' (\coqinline{repeat econstructor} on line 106) and
  the production of an initial execution state for the Clight program (lines
  107--114).
  \item Through application of an intermediate step
  (\coqinline{cbsem_func_sound}, see Lemma \ref{lem:cbsem_func_sound}), the
  second subgoal now becomes (in concise notation):
  \begin{align*}
    \vdash \coqAAAA{\vfCbsemFunc}{\coqT{[}\ \coqT{]}}{\vfxTrue}{\coqT{main_vf_func}}{\coqHole{?q}}.
  \end{align*}
  This subgoal can be proven by \coqinline{main_cbsem_by_symexec_tactics},
  our earlier lemma. Using that our notion of correctness of a Clight program
  does not say anything about the program's exit code, any postcondition $q$ is
  acceptable for proving the correctness of the Clight program, including the
  postcondition $\vfxBinop{\vfxVar{result}}{==}{\vfxIntLit{0}}$ verified in the
  above lemma.
\end{enumerate}

Since the Coq type checker accepts the VeriFast Coq and the CompCert libraries
and both scripts for the example program, we can be confident in our conclusion.
In fact, in the course of developing the soundness proof of \cbsem\ with respect
to Clight, we already discovered a bug in VeriFast itself, related to
appropriate checks on the result of a division (which was quickly fixed).


\section{\vfcx}
\label{section:vfcx}

We begin by describing the C language subset which is currently supported by the
export extension to VeriFast\footnote{The notation we use in this paper is more
compact than the notation we use in the actual Coq code. In the near future,
however, we will adapt our Coq code to use the notation introduced in this
paper.}. We chose just enough language features to be able to write programs
that diverge and that can exhibit some straightforward forms of undefined
behavior, specifically out-of-bounds arithmetic and division by zero. This means
that our subset has many notable omissions and limitations:
\begin{itemize}
\item we currently stick to variables having a single explicit type
(\vfcinline{int});
\item we do not support function calls;
\item all variables are locally allocated: there is no heap memory.
\end{itemize}

When VeriFast exports a parsed AST to Coq, the exported AST is actually a
program in a slightly different language, which we call
\emph{\vfcx}\footnote{This allows us to retain the name ``VeriFast C'' to refer
to VeriFast's C dialect: the collective set of choices which VeriFast in its
implementation makes with regards to the C language.}, with the \emph{x}
denoting ``export''. This language differs syntactically from actual C in a few
ways which will be quite familiar to many readers:
\begin{enumerate}
\item an OCaml/Haskell/... \texttt{let}-like variable binding is used for
variable declaration;
\item a \emph{list} of statements in C (e.g. \vfcinline{s1; s2; s3}) is
represented in \vfcx\ by nesting multiple instances of a binary \emph{sequence
operator} $\vfxSeq{\_}{\_}$ terminated by a $\vfxSkip$ statement (for the
earlier example we still get
$\vfxSeq{\vfxSeq{\vfxSeq{s_1}{s_2}}{s_3}}{\vfxSkip}$, because the sequence
operator notation is right associative);
\item the \vfcinline{while} loop carries its invariant directly within the
AST\footnote{In the Coq code, the \texttt{while} loop actually also has a second
statement which is currently not used.}.
\end{enumerate}
\vfcx\ does retain the use of named local variables (as opposed to common use of
De Bruijn indices).

The transformations of C into \vfcx\ make our work in Coq simpler, without
altering the meaning of programs substantially. The transformation into \vfcx\
takes place in VeriFast itself. We do not make any direct formal statement about
the relationship between the parsed C program and the transformed \vfcx\ program
which is exported to Coq\footnote{In Section \ref{section:clight} we will make 2
further transformations to the AST of a function. However, these transformations
were specifically geared towards Clight. Since we consider \cbsem\ as a result
standing on its own, we decided at the time to perform them within Coq rather
than within VeriFast.}. This means that all our theorems hold for \vfcx\ rather
than ``VeriFast C''.

This lack of formal relation is indirectly mitigated in section
\ref{section:clight}, where we will define a relation $\vfclEquivStmtSymbol$
between a function body $s$ exported by VeriFast and the corresponding function
body $s_\CL$ exported by CompCert's \texttt{clightgen} tool. Indeed, a crucial
element in proving correctness of a CompCert function with body $s_\CL$ by
symbolically executing $s$ will be our ability to prove that this relation holds
between $s$ and $s_\CL$.

\begin{definition}
  \label{def:expr}

  The set of \vfcx\ \emph{expressions} currently generated by our export code is
  defined inductively in Coq as:
  \begin{align*}
    e \in \vfExpr ::=\ & \vfxTrue \coqSep \vfxFalse \coqSep \vfxIntLit{z} \coqSep \vfxVar{x} \\
    & \coqSep \vfxBinop{e_1}{+}{e_2} \coqSep \vfxBinop{e_1}{-}{e_2} \coqSep \vfxBinop{e_1}{/}{e_2} \\
    & \coqSep \vfxBinop{e_1}{<}{e_2} \coqSep \vfxBinop{e_1}{\leq}{e_2} \coqSep \vfxBinop{e_1}{==}{e_2} \coqSep \vfxBinop{e_1}{\neq}{e_2} \\
    & \coqSep \vfxBinop{e_1}{\ \&\&\ }{e_2} \coqSep \vfxBinop{e_1}{\ ||\ }{e_2} \coqSep\ \vfxUnop{!}{e} \\
    & \coqSep \vfxBinop{e_1}{=}{e_2}.
  \end{align*}
\end{definition}

The above set includes boolean comparison operators, but we will
see how symbolic execution in Coq will limit usage of these operators in the
conditions of conditional statements and loops.

\begin{notation}
  $[\vfxVar{z}]$ denotes the actual Coq integer value $z \in \coqZ$
  corresponding to the integer literal $\vfxVar{z}$. Likewise, $[\vfxVar{x}]$
  denotes the actual Coq string value $x \in \coqString$ corresponding to
  variable reference $\vfxVar{x}$.
\end{notation}

\begin{definition}
  \label{def:stmt}

  Likewise, the set of \vfcx\ \emph{statements} currently generated by our
  export code is defined inductively as:
  \begin{align*}
    s \in \coqT{stmt} ::=\ & \vfxSkip \\
    & \coqSep \vfxSeq{s_1}{s_2} && \text{Sequence statement} \\
    & \coqSep \vfxLet{\vfxVar{x}}{e}{s} && \text{\texttt{let}-style variable declaration}\\
    & \coqSep e && \text{Expression statement} \\
    & \coqSep \vfxIf{e}{s_1}{s_2} \\
    & \coqSep \vfxReturn{e} \\
    & \coqSep \vfxWhile{e_c}{e_i}{s} && \text{Loop with condition $e_c$ and invariant $e_i$}\\
    & \coqSep \vfxBlock{s}. && \text{Block statement}
  \end{align*}
\end{definition}

All \vfcx\ \emph{functions} currently have return type \vfcinline{int} and take
zero or more arguments of type \vfcinline{int}. A \vfcx\ function in Coq is
defined by four things: (1) a list of argument names $\overline{x} \in \coqList\
\coqString$; (2) an expression $e_p$ representing the function's precondition;
(3) the function body $s$ and (4) an expression $e_q$ representing the
postcondition. We do not currently define a separate Coq type for a function,
but just pass around these four items as separate arguments wherever they are
needed.


\section{Symbolic execution in Coq}
\label{section:symexec}

This section discusses our approach to symbolic execution of a function in Coq.
This approach consists of two steps. Subsection \ref{subsection:sep} shows how
the function's AST, together with its pre- and postcondition, is first
transformed into a logical sentence, called the symbolic execution proposition
or SEP. This SEP is a Coq term of type $\coqProp$ which corresponds structurally
to the tree formed by VeriFast's internal branching and pushing and popping of
SMT solver contexts. It is a shallow embedding of the assertions made in the
various execution branches, each execution branch having its own set of
assumptions representing the path condition of that branch. Subsection
\ref{subsection:proving_sep} then describes how the VeriFast extension currently
exports a proof for this proposition.

\subsection{Constructing the SEP}
\label{subsection:sep}

We will provide our technical description of constructing a SEP in a rather
bottom-up fashion, starting from a few conceptual ideas and lower-level
definitions and working up towards the main correctness predicate for symbolic
execution of a \vfcx\ function: $\vfSymExecFunc$ (see definition
\ref{def:sym_exec_func} below). Coq of course works in the opposite direction by
gradually unfolding and reducing an instance of $\vfSymExecFunc$ into a term of
type $\coqProp$.

\subsubsection{Stores, continuations and symbol picking}

\begin{definition}
  \label{def:store}

  Given that we currently only support variables of type \vfcinline{int}, a
  \emph{symbolic store} is simply defined as a total function with an empty
  store $\vfEmptyStore$ and a single update operation $\vfStoreUpdate{\_}{\_}:
  \coqString \to \coqOption \coqZ \to \vfStore \to \vfStore$, handling both
  inserts and deletions:
  \begin{align*}
    \sigma \in \vfStore := & \coqString \to \coqOption \coqZ. \\
    \vfEmptyStore := & \coqL{\_}{\coqNone}. \\
    \vfStoreUpdate{x}{v} \sigma := & \coqL{y}{
      \begin{cases}
        v & \text{if $\coqCompBinop{x}{=}{y}$,} \\
        \sigma\ y & \text{otherwise.}
      \end{cases}
    }
  \end{align*}
\end{definition}

\emph{Continuations} allow us to express construction of the SEP by combining
various predicates in a \emph{continuation-passing style}, which corresponds
structurally well to the style of programming found in the OCaml source code of
VeriFast.
\begin{definition}[Continuations taking \emph{only} a store]
  \label{def:cont}
  \begin{align*}
    C \in \vfCont :=\ & \vfStore \to \coqProp.
  \end{align*}
\end{definition}
$\vfCont$ has been given its own type because it is so common throughout the
development. Other, less common continuations will take additional parameters on
top of a store. For instance, a \emph{return continuation} will have type $\coqZ
\to \vfCont$. In addition to a store, a return continuation also takes a return
value $z$ to generate a SEP. Important instances of these continuation
passing-style SEP-producing functions are the functions to pick fresh symbols.

\begin{definition}
  \label{def:for_Z}

  The function $\vfForZ: \coqString \to \vfCont \to \vfCont$ is defined as:
  \begin{align*}
    \vfMinSigned: \coqZ :=\ & -2147483648. \\
    \vfMaxSigned: \coqZ :=\ & 2147483647. \\
    \coqA{\coqT{is_int}}{z} :=\ & \vfMinSigned \leq z \leq \vfMaxSigned. \\
    \coqAAA{\vfForZ}{x}{C}{\sigma} :=\ & \forall\ z,\ \coqT{is_int}\ z \to C\ (\vfStoreUpdate{x}{\coqSome z}\ \sigma).
  \end{align*}
  $\vfForZ$ takes a variable identifier $x$ and, through a continuation,
  generates a proposition quantifying over $z \in \coqZ$, such that $z$ is
  within hardcoded integer bounds\footnote{In the future, we will allow
  architectural parametrization in our exported Coq script, similar to that
  found in CompCert.}.

  The fixpoint function $\vfForZs$ lifts $\vfForZ$ to pick fresh symbols for a
  whole list of identifiers $\overline{x} \in \coqList\ \coqString$ at once.
  Likewise, the fixpoint function $\vfHavocZs$ picks fresh symbols for a list of
  identifiers $\overline{x}$, \emph{with the added condition that the
  identifiers were already bound in the store which is passed to $\vfHavocZs$}.
  This condition does not exist for $\vfForZs$. We omit definitions for both
  functions.
\end{definition}

\subsubsection{Evaluation and translation of expressions}

An expression $e \in \vfExpr$ can be shallowly embedded in a SEP in two distinct
ways. When \emph{evaluating} an expression in the context of actual symbolic
execution, we require extra conditions to be embedded, for instance to check
integer bounds or to check for division by zero. These extra verification
conditions correspond to VeriFast making calls to the SMT solver during
execution of some C function. However, when \emph{translating} an expression in
the context of embedding pre- and postconditions and invariants in the SEP,
there is no need to add any extra checks in the SEP. Rather, the expression
\emph{is itself} the assumption made by the SMT solver or the assertion verified
by the SMT solver.

Since our concern with this work was to get a simple proof of concept working,
the total set of expressions which is handled properly by our Coq code is a
limited subset of what is allowed in C. To this end, we make a second, rather
artificial distinction between \emph{arithmetic expressions} ($\vfxIntLit{z}$,
$\vfxVar{x}$, $\vfxBinop{e_1}{+}{e_2}$, $\vfxBinop{e_1}{-}{e_2}$,
$\vfxBinop{e_1}{/}{e_2}$) and \emph{``Boolean'' expressions}
($\vfxBinop{e_1}{<}{e_2}$, $\vfxBinop{e_1}{\leq}{e_2}$,
$\vfxBinop{e_1}{==}{e_2}$, $\vfxBinop{e_1}{\neq}{e_2}$, $\vfxBinop{e_1}{\ \&\&\
}{e_2}$, $\vfxBinop{e_1}{\ ||\ }{e_2}$, $\vfxUnop{!}{e}$, $\vfxTrue$ and
$\vfxFalse$). The assignment expression $\vfxBinop{e_1}{=}{e_2}$ forms a third
category of its own. To enable further simplifications, only arithmetic and
Boolean expression are evaluated. Both are evaluated as
r-expressions\footnote{As will be seen in section
\ref{subsubsection:symexecstmt}, we don't need evaluation of l-expressions for
now.}.

\begin{definition}
  \label{def:eval_Z_cps}

  Fixpoint function $\vfEvalZCps: \vfExpr \to (\coqZ \to \vfCont) \to \vfCont$
  attempts to evaluate an expression $e$ to a Coq value $z$. The superscript
  $\coqT{c}$ denotes that it is a CPS function: if successful, a continuation
  $C$ is called with value $z$, producing a sub-SEP which is conjoined with any
  side conditions arising from the evaluation. That is,
  $\coqAAA{\vfEvalZCps}{e}{C}{\sigma}$ is true if and only if in store $\sigma$,
  $e$ has no undefined behavior and evaluates to a value $z$ such that
  $\coqAA{C}{z}{\sigma}$ is true.

  For integer literals and variable references, this evaluation is rather
  simple:
  \begin{align*}
    \coqAAA{\vfEvalZCps}{\vfxIntLit{z}}{C}{\sigma} := &
      \begin{cases}
        C\ [\vfxIntLit{z}]\ \sigma & \text{if $\coqCompBound{\vfMinSigned}{[\vfxIntLit{z}]}{\vfMaxSigned$},} \\
        \coqFalse & \text{otherwise.}
      \end{cases} \\
    \coqAAA{\vfEvalZCps}{\vfxVar{x}}{C}{\sigma} := &
      \begin{cases}
        C\ z\ \sigma & \text{if $\coqCompBinop{\sigma\ [\vfxIntLit{x}]}{=}{\coqSome z}$,} \\
        \coqFalse & \text{otherwise.}
      \end{cases}
  \end{align*}
  Evaluation of the arithmetic expression $\vfxBinop{\_}{+}{\_}$ is
  defined as:
  \begin{align*}
    \coqAAA{\vfEvalZCps}{(\vfxBinop{e_1}{+}{e_2})}{C}{\sigma} :=\ & \coqAAA{\vfEvalZCps}{e_1}{(\coqL{z_1\ \sigma}{ \\
    & \coqs{1}   \coqAAA{\vfEvalZCps}{e_2}{(\coqL{z_2\ \sigma}{ \\
    & \coqs{2}     \coqLet{z}{z_1 + z_2}{ \\
    & \coqs{2}     (\vfMinSigned \leq z) \land ((z \leq \vfMaxSigned) \land (\coqAA{C}{z}{\sigma}))}} \\
    & \coqs{1}   )}{\sigma}} \\
    & \coqs{0} )}{\sigma}.
  \end{align*}
  We omit the definition for $\vfxBinop{\_}{-}{\_}$, which is very
  similar to that for addition. For division we have:
  \begin{align*}
    \coqAAA{\vfEvalZCps}{(\vfxBinop{e_1}{/}{e_2})}{C}{\sigma} :=\ & \coqAAA{\vfEvalZCps}{e_1}{(\coqL{z_1\ \sigma}{ \\
    & \coqs{1}   \coqAAA{\vfEvalZCps}{e_2}{(\coqL{z_2\ \sigma}{ \\
    & \coqs{2}     \coqLet{z}{z_1 \div z_2}{ \\
    & \coqs{2}     z_2 \neq 0 \land ((z_1 \neq \vfMinSigned) \lor (z_2 \neq -1)) \land (\coqAA{C}{z}{\sigma})}} \\
    & \coqs{1}   )}{\sigma}} \\
    & \coqs{0} )}{\sigma}.
  \end{align*}
  Finally, for all other cases (Boolean expressions and the assignment
  expression), $\vfEvalZCps$ generates $\coqFalse$:
  \begin{align*}
    \coqAAA{\vfEvalZCps}{\_}{C}{\sigma} := \coqFalse.
  \end{align*}
  In principle, we could have simplified $\vfEvalZCps$, based on the observation
  that the current set of supported expressions does not change $\sigma$ during
  evaluation. This will however not remain the case in the future.
\end{definition}

The important idea is that $\vfEvalZCps$ tightly mirrors what happens in
VeriFast itself. It embeds verification conditions whenever VeriFast performs
SMT calls. Even the \emph{order} of the embedded verification conditions
corresponds to the order in which VeriFast calls the SMT solver: this can be
seen in the case for $+$, where first a comparison against $\vfMinSigned$ is
asserted, followed by a comparison against $\vfMaxSigned$, followed by the
continuation. Likewise, for the direct computational checks in the cases of
$\vfxIntLit{z}$ and $\vfxVar{x}$, VeriFast correspondingly performs the same
checks directly in the OCaml code (thereby \emph{avoiding} calls to the SMT
solver).

\begin{definition}
  \label{def:eval_Prop_cps}

  Fixpoint function $\vfEvalPropCps: \vfExpr \to (\coqProp \to \vfCont) \to
  \vfCont$ transforms the Boolean expressions into a SEP. Let us begin with the
  simple definitions for $\vfxTrue$ and $\vfxFalse$:
  \begin{align*}
    \coqAAA{\vfEvalPropCps}{\vfxTrue}{C}{\sigma} :=\ & \coqAA{C}{\coqTrue}{\sigma}. \\
    \coqAAA{\vfEvalPropCps}{\vfxFalse}{C}{\sigma} :=\ & \coqAA{C}{\coqFalse}{\sigma}.
  \end{align*}
  The cases for $\vfxBinop{\_}{==}{\_}$, $\vfxBinop{\_}{<}{\_}$,
  $\vfxBinop{\_}{\leq}{\_}$ and $\vfxBinop{\_}{\neq}{\_}$ are very similar and
  rest on calls to $\vfEvalZCps$ (Definition \ref{def:eval_Z_cps}), meaning the
  resulting proposition may contain side conditions. We only give the case for
  equality testing:
  \begin{align*}
    \coqAAA{\vfEvalPropCps}{(\vfxBinop{e_1}{==}{e_2})}{C}{\sigma} :=\ & \coqAAA{\vfEvalZCps}{e_1}{(\coqL{z_1\ \sigma}{\\
    & \coqs{1}   \coqAAA{\vfEvalZCps}{e_2}{(\coqL{z_2\ \sigma}{ \\
    & \coqs{2}     C\ (z_1 = z_2)\ \sigma} \\
    & \coqs{1}   )}{\sigma}} \\
    & \coqs{0} )}{\sigma}.
  \end{align*}
  The cases for the logical binary operations are likewise very similar. We only
  provide a definition for logical conjunction:
  \begin{align*}
    \coqAAA{\vfEvalPropCps}{(\vfxBinop{e_1}{\land}{e_2})}{C}{\sigma} :=\ & \coqAAA{\vfEvalPropCps}{e_1}{(\coqL{E_1\ \sigma}{\\
    & \coqs{1}   \coqAAA{\vfEvalPropCps}{e_2}{(\coqL{E_2\ \sigma}{ \\
    & \coqs{2}     C\ (E_1 \land E_2)\ \sigma} \\
    & \coqs{1}   )}{\sigma}} \\
    & \coqs{0} )}{\sigma}.
  \end{align*}
  The definition for logical negation is given as:
  \begin{align*}
    \coqAAA{\vfEvalPropCps}{(\vfxUnop{!}{e})}{C}{\sigma} := \coqAAA{\vfEvalPropCps}{e}{(\coqL{E\ \sigma}{\coqAA{C}{(\neg E)}{\sigma}})}{\sigma}.
  \end{align*}
  All other cases (assignment expressions and arithmetic expressions) currently
  evaluate to $\coqFalse$.
\end{definition}

\begin{definition}
  \label{def:translate_expr_to_Prop_cps}

  Fixpoint function $\vfTranslateProp: \vfExpr \to (\coqProp \to \coqProp) \to
  \vfCont$ \emph{translates} expression $e$ in store $\sigma$ into $E \in
  \coqProp$, such that $\coqAAA{\vfTranslateProp}{e}{C}{\sigma}$ is true if
  $\coqAA{C}{E}{\sigma}$ is true.

  It differs from $\vfEvalPropCps$ in that no side conditions for division by
  zero and integer boundedness are generated for the cases involving arithmetic.
  The exact definition for $\vfTranslateProp$ is not provided here.
\end{definition}

\begin{definition}
  \label{def:produce_consume}

  The functions $\vfProduce: \vfExpr \to \vfCont \to \vfCont$ and $\vfConsume:
  \vfExpr \to \vfCont \to \vfCont$ are defined as:
  \begin{align*}
    \coqAAA{\vfProduce}{e}{C}{\sigma} :=\ & \coqAAA{\vfTranslateProp}{e}{(\coqL{E}{E \to \coqA{C}{\sigma}})}{\sigma}. \\
    \coqAAA{\vfConsume}{e}{C}{\sigma} :=\ & \coqAAA{\vfTranslateProp}{e}{(\coqL{E}{E\ \coqAnd\ \coqA{C}{\sigma}})}{\sigma}.
  \end{align*}
  These functions leverage $\vfTranslateProp$ for a common use case: the actual
  \emph{production} (assumption) of pre-conditions and \emph{consumption}
  (assertion) of post-conditions.
\end{definition}

\subsubsection{Symbolic execution of statements}
\label{subsubsection:symexecstmt}

Before we can get to constructing a SEP which embeds the symbolic execution of
statements, we need to introduce two auxiliary functions.

\begin{definition}
  \label{def:stmt_to_free_possible_targets}
  Fixpoint function $\vfFreeTargets: \vfStmt \to \coqList\ \coqString$ collects
  the names of all free variables that are \emph{potentially} modified by some
  statement $s$. We omit definition of this function.
\end{definition}

\begin{definition}
  \label{def:leak_check}

  Function $\vfLeakCheck: \vfCont$ is defined as:
  \begin{align*}
    \coqA{\vfLeakCheck} := \coqL{\_}{\coqTrue}.
  \end{align*}
  It is a placeholder SEP generating function which corresponds to the checks
  performed by VeriFast upon exiting a function or at the end of a loop body.
  Since we do not deal with heap memory in this proof of concept, its
  implementation remains trivial for now.
\end{definition}

\begin{definition}
  \label{def:sym_exec_stmt}

  Symbolic execution of a \emph{statement} $s$ is implemented by the fixpoint
  function $\vfSymExecStmt: \vfStmt \to \vfCont \to (\coqZ \to \vfCont) \to
  \vfCont$. This function takes two continuations, which correspond to the two
  possible outcomes for executing a statement: \emph{normal continuation} ($C_N
  \in \vfCont$) and \emph{returning from the function} ($C_R \in \coqZ \to
  \vfCont$):
  \begin{align*}
    \coqAAAA{\vfSymExecStmt}{\vfxSkip}{C_N}{C_R}{\sigma} :=\ & \coqA{C_N}{\sigma}. \\
    \coqAAAA{\vfSymExecStmt}{(\vfxSeq{s_1}{s_2})}{C_N}{C_R}{\sigma} :=\ & \coqAAAA{\vfSymExecStmt}{s_1}{(\coqAAA{\vfSymExecStmt}{s_2}{C_N}{C_R})}{C_R}{\sigma}. \\
    \coqAAAA{\vfSymExecStmt}{(\vfxReturn{e})}{C_N}{C_R}{\sigma} :=\ & \coqAAA{\vfEvalZCps}{e}{C_R}{\sigma}. \\
    \coqAAAA{\vfSymExecStmt}{\vfxBlock{s}}{C_N}{C_R}{\sigma} :=\ & \coqAAAA{\vfSymExecStmt}{s}{C_N}{C_R}{\sigma}.
  \end{align*}
  For local variable declarations, we have:
  \begin{align*}
    & \coqAAAA{\vfSymExecStmt}{(\vfxLet{x}{e}{s})}{C_N}{C_R}{\sigma} := \\
    & \coqs{5}   \coqIf{(\coqA{\sigma}{x})}{\coqFalse}{\\
    & \coqs{6}     \coqAAA{\vfEvalZCps}{e}{(\coqL{z\ \sigma}{\\
    & \coqs{7}       \coqLet{C'_N}{(\coqL{\sigma}{\coqA{C_N}{\vfStoreUpdate{x}{\coqNone}\sigma}})}{\\
    & \coqs{7}       \coqLet{C'_R}{(\coqL{z\ \sigma}{\coqAA{C_R}{z}{\vfStoreUpdate{x}{\coqNone}\sigma}})}{\\
    & \coqs{7}       \coqAAAA{\vfSymExecStmt}{s}{C'_N}{C'_R}{\vfStoreUpdate{x}{\coqA{\coqSome}{z}}\sigma} \\
    & \coqs{6}     }}})} {\sigma}}.
  \end{align*}
  In terms of expression statements, our symbolic execution currently only
  supports direct assignments of an expression $e$ to a variable identifier
  expression $x$, which is why in practice we only need to care about evaluation
  of r-expressions:
  \begin{align*}
    \coqAAAA{\vfSymExecStmt}{(\vfxBinop{\vfxVar{x}}{=}{e})}{C_N}{C_R}{\sigma} :=\ & {
      \coqAAA{\vfEvalZCps}{e}{(\coqL{z\ \sigma}{\coqA{C_N}{\vfStoreUpdate{[\vfxVar{x}]}{\coqA{\coqSome}{z}}\sigma}})}{\sigma}
    }.
  \end{align*}
  Conditional statements branch into two path conditions, one in which the
  condition $e \in \vfExpr$ is true and one in which it is not true:
  \begin{align*}
    \coqAAAA{\vfSymExecStmt}{(\vfxIf{e}{s_1}{s_2})}{C_N}{C_R}{\sigma} :=\ & \coqAAA{\vfEvalPropCps}{e}{(\coqL{E\ \sigma}{ \\
    & \coqs{1}   (E \to (\coqAAAA{\vfSymExecStmt}{s_1}{C_N}{C_R}{\sigma})) \\
    & \coqs{1}   \land \\
    & \coqs{1}   (\neg E \to (\coqAAAA{\vfSymExecStmt}{s_2}{C_N}{C_R}{\sigma})) \\
    & \coqs{0} )}}{\sigma}.
  \end{align*}
  Finally, \vfcinline{while} loops are a little bit more involved:
  \begin{align*}
    \coqAAAA{\vfSymExecStmt}{&(\vfxWhile{e_c}{e_i}{s})}{C_N}{C_R}{\sigma} := \coqAAA{\vfConsume}{e_i}{( \\
    & \coqs{0} \coqAA{\vfHavocZs}{(\coqA{\vfFreeTargets}{s})}{( \\
    & \coqs{1}   \coqAA{\vfProduce}{e_i}{( \\
    & \coqs{2}     \coqAA{\vfEvalPropCps}{e_c}{(\coqL{E_c\ \sigma}{ \\
    & \coqs{3}       (E_C \to (\coqAAAA{\vfSymExecStmt}{s}{(\coqAA{\vfConsume}{e_i}{\vfLeakCheck})}{C_R}{\sigma})) \\
    & \coqs{3}       \land \\
    & \coqs{3}       (\neg E_C \to \coqA{C_N}{\sigma}) \\
    & \coqs{2}     })} \\
    & \coqs{1}   )}\\
    & \coqs{0} )})}{\sigma}.
  \end{align*}
  This chaining of continuations will unfold to reproduce the verification steps
  shown graphically in Figure \ref{fig:symexec}.
\end{definition}

\subsubsection{Symbolic execution of a function}

Before we can finally define symbolic execution of a function, we need to
examine one last predicate, $\vfRet$, which is used to generate the symbolic
store in which the postcondition can be asserted upon return.

\begin{definition}
  \label{def:ret}
  Predicate $\vfRet: \vfStore \to \vfCont \to \coqZ \to \vfCont$ \emph{retains}
  an ``old'' symbolic store $\sigma_0$ and ultimately passes on this store,
  updated with a return value $z$, to its own continuation $C$ (which will embed
  the actual postcondition, as shown below in definition
  \ref{def:sym_exec_func}):
  \begin{align*}
    \coqAAA{\vfRet}{\sigma_0}{C}{z} := \coqL{\_}{\coqA{C}{(\vfStoreUpdate{\coqT{"return"}}{\coqA{\coqSome}{z}} \sigma_0)}}.
  \end{align*}
  The continuation produced by $\vfRet$ simply discards the ``new'' store with
  which it is called. Indeed, a postcondition will be evaluated in terms of the
  original arguments provided to the function, together with the return value.
\end{definition}

\begin{definition}
  \label{def:sym_exec_func}

  The function correctness predicate $\vfSymExecFunc: \coqList\ \coqString \to
  \vfExpr \to \vfStmt \to \vfExpr \to \coqProp$ for a function with arguments
  $\overline{x}$, pre- and postconditions $e_p$ and $e_q$ and body $s$ is
  defined as:
  \begin{align*}
    \coqAAAA{\vfSymExecFunc}{\overline{x}}{e_p}{s}{e_q} :=\ & \vfForZs\ \overline{x}\ (\vfProduce\ e_p\ (\lambda\ \sigma, \\
    & \coqs{1} \coqA{\vfSymExecStmt}{s} \\
    & \coqs{2}   (\coqL{\_}{\coqFalse}) \\
    & \coqs{2}   (\coqAA{\vfRet}{\sigma}{(\coqAA{\vfConsume}{e_q}{\vfLeakCheck})}) \\
    & \coqs{2}   \sigma \\
    & \coqs{1} ) \\
    & ))\ \vfEmptyStore.
  \end{align*}
\end{definition}

Intuitively, definition \ref{def:sym_exec_func} states that nothing is known
about the function arguments, except that they are integers and that the
precondition must hold. Definition \ref{def:sym_exec_func} makes it necessary
for a function to either \emph{diverge} or, if it does terminate, to
\emph{return} an integer result $z$, the value of which is limited by a path
condition which satisfies the function's postcondition:

\begin{itemize}
\item A diverging symbolic execution path will accumulate contradictory
assumptions, the presence of which instantly terminates symbolic execution.
\item If some symbolic execution path of $s$ does not somehow end with a return
statement, the normalization of the subterm with $\vfSymExecStmt$ would end up
calling the normal continuation, which is $\coqL{\_}{\coqFalse}$. Proving the
resulting SEP would require proving $\coqFalse$, which is impossible.
\end{itemize}

\subsection{Proving the SEP}
\label{subsection:proving_sep}

After we have constructed and normalized a SEP for the function, we need to
provide a proof for it. Proof steps for the SEP in Coq are closely related to
changes in the VeriFast SMT solver state. We currently identify four types of
proof steps:

\begin{itemize}
  \item \emph{Assumptions}, for instance introduced by predicates $\vfForZs$ or
  $\vfProduce$, require us to move the variable or hypothesis into the proof
  context. In VeriFast this is accomplished by making \texttt{\#assume} calls to
  the SMT solver API.

  \item If the proof goal is a \emph{conjunction}, for instance produced by
  branching in some cases of $\vfSymExecStmt$ or by $\vfConsume$, we need to
  split it into two subgoals. In VeriFast, this is implemented by a function
  \texttt{\#branch}, which pushes a new SMT solver context for each branch,
  popping the context after the branch finishes.

  \item \emph{Assertions} are generated by e.g. $\vfConsume$ or some
  continuations. They are handled in VeriFast by \texttt{\#assert} calls to the
  SMT solver.

  \item As mentioned earlier, a path condition may accumulate
  \emph{contradictory assumptions}, for instance in the context of a diverging
  execution. At that point we need to point out these contradictions discharging
  the proof goal. In VeriFast, the SMT solver's internal mechanism detects this.
\end{itemize}

In Table~\ref{tab:proving} we summarize our current approach to implementing
these four types of proof steps in Coq. Since we may benefit from the limited,
integer arithmetic nature of the current set of expressions which has to be
supported, a second option would be to make a single tactic which repeatedly
applies $\coqT{intros}$, $\coqT{split}$ and $\coqT{lia}$ from within a
syntax-driven goal match.

\begin{table}[]
  \centering
  \begin{tabular}{p{30mm}|p{30mm}}
    & Tactic \\
    \hline

    Assumptions
      & $\coqT{intro}$
      \\

    Conjunctions
      & $\coqT{split}$
      \\

    Assertions
      & $\coqT{lia}$
      \\

    Contradictions
      & $\coqT{lia}$
  \end{tabular}
  \caption{Four types of proof steps.}
  \label{tab:proving}
\end{table}

For now, we choose to make VeriFast export a proof. We record all calls to the
SMT API and the \texttt{\#branch} function and use these recordings to generate
a tactics proof in the exported proof. The main difficulty in recording so far
was that contradictions ``break out'' of the SMT solver, meaning they are not
explicitly recorded. Contradictions can be seen in the recordings as the
\emph{absence} of an assertion or equivalently, as a ``dangling'' assumption.
Since contradictions in our case stem from assumptions about integer
expressions, we also use $\coqT{lia}$ to deal with these.

Exporting the proof is at this point not only more cumbersome than a simple
syntax-driven tactic in Coq, it is also quite \emph{fragile}. Ironically, this
fragility can be seen as an advantage. It can point out bugs, in VeriFast or in
the Coq code, when both do different things. If the SEP generated in Coq misses
an assertion or assumption, or if it has too many assertions or assumptions, the
exported proof may not line up with the SEP.


\section{Big step semantics for \vfcx}
\label{section:cbsem}

Before we prove the soundness of VeriFast's symbolic execution with respect to
Clight, let us first look at an intermediate step. In this section, we develop
\cbsem, a formal semantics for \vfcx.

We chose to express ourselves using big step semantics (see Subsection
\ref{subsection:cbsem}), because this approach was somewhat easier to get
started with. The downside is that our inductive big step semantics, seen by
itself, does not allow us to distinguish program executions that have undefined
behavior from those that diverge. Neither type of execution outcome can be
derived in this big step semantics.

At first we solved this problem by adding a timeout counter to each semantic
rule. This worked well in itself and would have allowed us to stick to a single
set of derivation rules. However, in order to prove soundness with regards to
Clight, we decided it would be easier to reflect Clight's approach to big step
semantics \cite{blazy-clight-jar-2009, leroy-cbsos-ic-2009} by \emph{adding} a
separate, \emph{coinductive} set of semantic rules for diverging programs (see
Subsection \ref{subsection:cbseminf}). \cbsem\ therefore refers to the
\emph{combined} inductive and coinductive sets of semantic rules.

Throughout this section, we will relate symbolic execution as specified in
Section \ref{section:symexec} to the inductive and the coinductive semantics
that are being developed, again building up from expressions to functions.
Finally, in Subsection \ref{subsection:symexec_cbsem} we provide a
\emph{soundness theorem} proving that symbolic execution of a function is sound
with regards to our formal big step semantics \cbsem.

\subsection{Terminating statements}
\label{subsection:cbsem}

\begin{definition}
  \label{def:outcome}

  The set $\vfOut$ specifies the \emph{possible outcomes} of executing some
  \vfcx\ statement. It is defined inductively as:
  \begin{align*}
    o \in \vfOut ::=\ & \vfOutN\ |\ \coqA{\vfOutR}{z}.
  \end{align*}
  In the same way that predicate $\vfSymExecStmt$ for symbolic execution accepts
  two continuations, one for a normal continuation and one for a return
  situation, $\vfOut$ is defined by two constructors.
\end{definition}

Before we can get to the actual inference rules for our big step semantics, we
need to mention variations of $\vfEvalZCps$ and $\vfEvalPropCps$ which do
\emph{not} use the continuation passing style.

\begin{definition}
  \label{def:eval_Z_bool}

  Fixpoint function $\vfEvalZ: \vfExpr \to \vfStore \to \coqOption\ \coqZ$
  attempts to evaluate an expression $e$ to a value $z$. This method will
  evaluate in exactly the same way as the method $\vfEvalZCps$ from Definition
  \ref{def:eval_Z_cps}. The lack of the superscript $\coqT{c}$ denotes that it
  directly returns $z$ (as an $\coqOption$ value), instead of passing it to some
  continuation.

  Likewise, fixpoint function $\vfEvalBool: \vfExpr \to \vfStore \to \coqOption\
  \coqBool$ evaluates expression $e$ in a way comparable to $\vfEvalPropCps$
  from Definition \ref{def:eval_Prop_cps}; but it computes an actual boolean
  value, instead of a logical proposition and does not pass it on to a
  continuation.

  We omit precise definitions for both functions.
\end{definition}

\begin{definition}
  \label{def:cbsem_exec_stmt}

  The big step relation $\vfCbsemExecStmt{\sigma}{s}{\sigma'}{o}$ describes the
  \emph{finite execution} of statement $s$ starting from store $\sigma$,
  resulting in store $\sigma'$ with outcome $o$. It is inductively defined as:
  \begin{mathparpagebreakable}
    \inferrule
      {\\}{\vfCbsemExecStmt{\sigma}{\vfxSkip}{\sigma}{\vfOutN}}
    \and
    \inferrule
      {
        \vfCbsemExecStmt{\sigma}{s_1}{\sigma'}{\vfOutN} \\
        \vfCbsemExecStmt{\sigma'}{s_2}{\sigma''}{o}
      }{
        \vfCbsemExecStmt{\sigma}{(\vfxSeq{s_1}{s_2})}{\sigma''}{o}
      }
    \and
    \inferrule
      {
        \vfCbsemExecStmt{\sigma}{s_1}{\sigma'}{o} \\
        o \neq \vfOutN
      }{
        \vfCbsemExecStmt{\sigma}{(\vfxSeq{s_1}{s_2})}{\sigma'}{o}
      }
    \and
    \inferrule
      {
        \coqA{\sigma}{x} = \coqNone \\
        \coqAA{\vfEvalZ}{e}{\sigma} = \coqA{\coqSome}{z} \\
        \vfCbsemExecStmt{\vfStoreUpdate{x}{\coqA{\coqSome}{z}} \sigma}{s}{\sigma'}{o}
      }{
        \vfCbsemExecStmt{\sigma}{(\vfxLet{x}{e}{s})}{\vfStoreUpdate{x}{\coqNone} \sigma'}{o}
      }
    \and
    \inferrule
      {
        \coqAA{\vfEvalZ}{e}{\sigma} = \coqA{\coqSome}{z}
      }{
        \vfCbsemExecStmt{\sigma}{(\vfxBinop{\vfxVar{x}}{=}{e})}{\vfStoreUpdate{[\vfxVar{x}]}{\coqA{\coqSome}{z}} \sigma}{\vfOutN}
      }
    \and
    \inferrule
      {
        \coqAA{\vfEvalBool}{e}{\sigma} = \coqA{\coqSome}{\coqBoolTrue} \\
        \vfCbsemExecStmt{\sigma}{s_1}{\sigma'}{o}
      }{
        \vfCbsemExecStmt{\sigma}{(\vfxIf{e}{s_1}{s_2})}{\sigma'}{o}
      }
    \and
    \inferrule
      {
        \coqAA{\vfEvalBool}{e}{\sigma} = \coqA{\coqSome}{\coqBoolFalse} \\
        \vfCbsemExecStmt{\sigma}{s_2}{\sigma'}{o}
      }{
        \vfCbsemExecStmt{\sigma}{(\vfxIf{e}{s_1}{s_2})}{\sigma'}{o}
      }
    \and
    \inferrule
      {
        \coqAA{\vfEvalZ}{e}{\sigma} = \coqA{\coqSome}{z}
      }{
        \vfCbsemExecStmt{\sigma}{(\vfxReturn{e})}{\sigma}{\coqA{\vfOutR}{z}}
      }
    \and
    \inferrule
      {
        \coqAA{\vfEvalBool}{e_c}{\sigma} = \coqA{\coqSome}{\coqBoolTrue} \\
        \vfCbsemExecStmt{\sigma}{s}{\sigma'}{o} \\
        o \neq \vfOutN
      }{
        \vfCbsemExecStmt{\sigma}{(\vfxWhile{e_c}{e_i}{s})}{\sigma'}{o}
      }
    \and
    \inferrule
      {
        \coqAA{\vfEvalBool}{e_c}{\sigma} = \coqA{\coqSome}{\coqBoolTrue} \\
        \vfCbsemExecStmt{\sigma}{s}{\sigma'}{\vfOutN} \\
        \vfCbsemExecStmt{\sigma'}{(\vfxWhile{e_c}{e_i}{s})}{\sigma''}{o}
      }{
        \vfCbsemExecStmt{\sigma}{(\vfxWhile{e_c}{e_i}{s})}{\sigma''}{o}
      }
    \and
    \inferrule
      {
        \coqAA{\vfEvalBool}{e_i}{\sigma} = \coqA{\coqSome}{\coqBoolFalse}
      }{
        \vfCbsemExecStmt{\sigma}{(\vfxWhile{e_c}{e_i}{s})}{\sigma}{\vfOutN}
      }
    \and
    \inferrule
      {
        \vfCbsemExecStmt{\sigma}{s}{\sigma'}{o}
      }{
        \vfCbsemExecStmt{\sigma}{\vfxBlock{s}}{\sigma'}{o}
      }
  \end{mathparpagebreakable}
\end{definition}

We can now proceed to explore the relation between symbolic execution and our
semantics for terminating statements.

\begin{definition}
  \label{def:well_typed}

  The \emph{well-boundedness} predicate $\vfWellTyped: \vfStore \to \coqProp$ is
  defined as:
  \begin{align*}
    \coqA{\vfWellTyped}{\sigma} := \forall\ x\ z,\ \coqA{\sigma}{x} = \coqA{\coqSome}{z} \to \coqA{\vfIsInt}{z}.
  \end{align*}
  It asserts that all values in some store $\sigma$ have a value $z$ within
  integer bounds $\vfMinSigned \leq z \leq \vfMaxSigned$ (see also Definition
  \ref{def:for_Z}).
\end{definition}

\begin{lemma}
  \label{lem:symexec_stmt_sound_termination}

  Given a store $\sigma$ such that $\coqA{\vfWellTyped}{\sigma}$, if symbolic
  execution of some statement $s$ succeeds
  ($\coqAAAA{\vfSymExecStmt}{s}{C_N}{C_R}{\sigma}$) and a terminating execution
  can be derived in our semantics ($\vfCbsemExecStmt{\sigma}{s}{\sigma'}{o}$),
  then:
  \begin{align*}
    & \coqA{C_N}{\sigma'} && \text{if $o = \vfOutN$,} \\
    & \coqAA{C_R}{z}{\sigma'} && \text{if $o = \coqA{\vfOutR}{z}$.}
  \end{align*}
  This lemma makes it possible to use the outcome $o$ from an execution in
  $\vfCbsemExecStmtSymbol$ as a witness for the validity of either the normal or
  returning continuation in the equivalent symbolic execution.
\end{lemma}
\begin{proof}
  The while loop requires us to proceed by induction on the derivation of
  $\vfCbsemExecStmt{\sigma}{s}{\sigma'}{o}$.
\end{proof}

\subsection{Diverging statements}
\label{subsection:cbseminf}

\begin{definition}
  \label{def:cbsem_execinf_stmt}

  The big step relation $\vfCbsemExecinfStmt{\sigma}{s}$ describes the
  \emph{diverging execution} of statement $s$ starting from store $\sigma$.
  Since it does not terminate, it is not associated with a final state or an
  outcome. It is coinductively defined as:

  \begin{mathparpagebreakable}
    \mprset{fraction={===}}
    \inferrule{
      \vfCbsemExecinfStmt{\sigma}{s_1}
    }{
      \vfCbsemExecinfStmt{\sigma}{(\vfxSeq{s_1}{s_2})}
    }
    \and
    \inferrule{
      \vfCbsemExecStmt{\sigma}{s_1}{\sigma'}{\vfOutN} \\
      \vfCbsemExecinfStmt{\sigma'}{s_2}
    }{
      \vfCbsemExecinfStmt{\sigma}{(\vfxSeq{s_1}{s_2})}
    }
    \and
    \inferrule{
      \coqA{\sigma}{x} = \coqNone \\
      \coqAA{\vfEvalZ}{e}{\sigma} = \coqA{\coqSome}{z} \\
      \vfCbsemExecinfStmt{\vfStoreUpdate{x}{\coqA{\coqSome}{z}} \sigma}{s}
    }{
      \vfCbsemExecinfStmt{\sigma}{(\vfxLet{x}{e}{s})}
    }
    \and
    \inferrule{
      \coqAA{\vfEvalBool}{e}{\sigma} = \coqA{\coqSome}{\coqBoolTrue} \\
      \vfCbsemExecinfStmt{\sigma}{s_1}
    }{
      \vfCbsemExecinfStmt{\sigma}{(\vfxIf{e}{s_1}{s_2})}
    }
    \and
    \inferrule{
      \coqAA{\vfEvalBool}{e}{\sigma} = \coqA{\coqSome}{\coqBoolFalse} \\
      \vfCbsemExecinfStmt{\sigma}{s_2}
    }{
      \vfCbsemExecinfStmt{\sigma}{(\vfxIf{e}{s_1}{s_2})}
    }
    \and
    \inferrule{
      \coqAA{\vfEvalBool}{e_c}{\sigma} = \coqA{\coqSome}{\coqBoolTrue} \\
      \vfCbsemExecinfStmt{\sigma}{s}
    }{
      \vfCbsemExecinfStmt{\sigma}{(\vfxWhile{e_c}{e_i}{s})}
    }
    \and
    \inferrule{
      \coqAA{\vfEvalBool}{e_c}{\sigma} = \coqA{\coqSome}{\coqBoolTrue} \\
      \vfCbsemExecStmt{\sigma}{s}{\sigma'}{\vfOutN} \\
      \vfCbsemExecinfStmt{\sigma'}{(\vfxWhile{e_c}{e_i}{s})}
    }{
      \vfCbsemExecinfStmt{\sigma}{(\vfxWhile{e_c}{e_i}{s})}
    }
    \and
    \inferrule{
      \vfCbsemExecinfStmt{\sigma}{s}
    }{
      \vfCbsemExecinfStmt{\sigma}{\vfxBlock{s}}
    }
  \end{mathparpagebreakable}
\end{definition}

We can now relate the coinductive relation $\vfCbsemExecinfStmtSymbol$ to
symbolic execution.

\begin{definition}
  \label{def:exec_stmt_not_derivable}

  Predicate $\vfCbsemExecStmtDerivable{\sigma}{s}$ expresses the \emph{ability}
  to derive a terminating execution for some statement $s$ starting from
  $\sigma$:
  \begin{align*}
    \vfCbsemExecStmtDerivable{\sigma}{s}\ := \exists\ \sigma'\ o,\ \vfCbsemExecStmt{\sigma}{s}{\sigma'}{o}.
  \end{align*}
  The \emph{inability} to derive a terminating execution, expressed $\neg
  (\vfCbsemExecStmtDerivable{\sigma}{s})$, may stem from either stuckness (for
  instance, in the case of division-by-zero) or from divergence.
\end{definition}

\begin{lemma}
  \label{lem:symexec_stmt_sound_divergence}

  Given a store $\sigma$ such that $\vfWellTyped{\sigma}$, if symbolic execution
  of some statement $s$ succeeds
  ($\coqAAAA{\vfSymExecStmt}{s}{C_N}{C_R}{\sigma}$) but a terminating execution
  cannot be derived using \cbsem\ (in other words, $\neg
  (\vfCbsemExecStmtDerivable{\sigma}{s})$), then we can derive a diverging
  execution $\vfCbsemExecinfStmt{\sigma}{s}$.
\end{lemma}
\begin{proof}
  The proof is coinductive and proceeds by case analysis on $s$. Many cases can
  be discharged \emph{ex falso} by construction of a terminating derivation,
  contradicting the assumption $\neg (\vfCbsemExecStmtDerivable{\sigma}{s})$.

  For the case of a while loop where the loop body $s'$ is executed, we need to
  \emph{decide} whether the loop body terminates from $\sigma$:
  $(\vfCbsemExecStmtDerivable{\sigma}{s'}) \lor \neg
  (\vfCbsemExecStmtDerivable{\sigma}{s'})$. Because of the halting problem, we
  are forced at this point to employ an axiomatic instance of the excluded
  middle, leading to two subcases corresponding to the two rules for diverging
  while loops in Definition \ref{def:cbsem_execinf_stmt}.
\end{proof}

Lemma \ref{lem:symexec_stmt_sound_divergence} states that if we cannot make a
terminating derivation for some $s$ starting from $\sigma$, then the evidence
provided by successful symbolic execution of $s$ allows us to narrow down
stuckness and divergence as causes for $\neg
(\vfCbsemExecStmtDerivable{\sigma}{s})$ to just divergence.

\subsection{Soundness of symbolic execution wrt \cbsem}
\label{subsection:symexec_cbsem}

\begin{definition}
  \label{def:store_eq_mod_xs}

  Two stores $\sigma$ and $\sigma'$ are \emph{equivalent modulo a list of
  identifiers $\overline{x}$}, written $\sigma \equiv_{\overline{x}} \sigma'$,
  if they satisfy $\forall\ x, (x \in \overline{x}) \lor (\coqA{\sigma}{x} =
  \coqA{\sigma'}{x})$.
\end{definition}

\begin{definition}
  \label{def:binds}

  Store $\sigma$ \emph{binds} variables $\overline{x}$, notation $\overline{x}
  \subseteq \coqA{\vfDom}{\sigma}$, if $\sigma\ x \neq \coqNone$ for $x \in
  \overline{x}$.
\end{definition}

Since we have implemented stores as functions, a list of variables bound by a
store must be passed around explicitly in various places. This is a bit of a
nuisance and will be fixed in future versions of this work, by representing
stores as finite maps rather than functions.

\begin{definition}
  \label{def:cbsem_exec_func_correct}

  The predicate $\vfCbsemFunc: \coqList\ \coqString \to \vfExpr \to \vfStmt \to
  \vfExpr \to \coqProp$ for a function with arguments $\overline{x}$, pre- and
  postconditions $e_p$ and $e_q$ and function body $s$ expresses correctness of
  that function with regards to \cbsem. Formally, this predicate asserts that:
  \begin{align*}
    & \coqs{0} \coqAAA{\vfTranslateProp}{e_p}{(\coqL{P}{ \\
    & \coqs{1}   P \to (\\
    & \coqs{2}     \exists\ \sigma,\ z, \coqAAA{\vfTranslateProp}{e_q}{(\coqL{Q}{(\vfCbsemExecStmt{\sigma}{s}{\sigma'}{\coqA{\vfOutR}{z}}) \land Q})}{\vfStoreUpdate{\coqT{"result"}}{z} \sigma} \\
    & \coqs{1}   ) \lor (\vfCbsemExecinfStmt{\sigma}{s}) \\
    & \coqs{0} })}{\sigma},
  \end{align*}
  for any store $\sigma$ such that $\coqA{\vfWellTyped}{\sigma}$, $\vfEmptyStore
  \equiv_{\overline{x}} \sigma$ and $\overline{x} \subseteq \coqA{\vfDom}{\sigma}$.
\end{definition}

$\vfCbsemFunc$ expresses that a function is correct with regards to
\cbsem\ if, given that the precondition holds for the call arguments, it
either terminates with a return outcome which satisfies the postcondition, or
diverges. It excludes any undefined behavior and it does not allow termination
with a normal outcome.

We will now conclude this section by proving soundness of the earlier notion of
function correctness from symbolic execution, $\vfSymExecFunc$, to
$\vfCbsemFunc$. The proof for this requires a similar soundness lemma for
statements.

\begin{lemma}[Soundness for statements]
  \label{lem:symexec_stmt_sound}

  If $\coqAAAA{\vfSymExecStmt}{s}{C_N}{C_R}{\sigma}$ holds with
  $\coqA{\vfWellTyped}{\sigma}$, then either:
  \begin{enumerate}
    \item
      $s$ terminates normally; that is, there exists some $\sigma'$ such that we
      can derive $\vfCbsemExecStmt{\sigma}{s}{\sigma'}{\vfOutN}$ and
      $\coqA{C_N}{\sigma'}$ holds;
    \item
      $s$ terminates by returning a value; that is, there exists some $\sigma'$
      and $z$ such that we can derive
      $\vfCbsemExecStmt{\sigma}{s}{\sigma'}{\coqA{\vfOutR}{z}}$ and
      $\coqAA{C_R}{z}{\sigma'}$ holds;
    \item
      $s$ diverges; that is, $\vfCbsemExecinfStmt{\sigma}{s}$.
  \end{enumerate}
\end{lemma}
\begin{proof}
  Similar to what we encountered in the proof of Lemma
  \ref{lem:symexec_stmt_sound_divergence} for the case of the while loop, the
  halting problem forces us to commence this proof with a case analysis on the
  axiomatic tautology $(\vfCbsemExecStmtDerivable{\sigma}{s'})\ \lor\ \neg
  (\vfCbsemExecStmtDerivable{\sigma}{s'})$. The terminating case proceeds by
  case analysis of the outcome of terminating execution, followed by application
  of Lemma \ref{lem:symexec_stmt_sound_termination}. The diverging case depends
  on Lemma \ref{lem:symexec_stmt_sound_divergence}.
\end{proof}

\begin{theorem}[Soundness for functions]
  \label{thm:symexec_func_sound}

  Given a function with arguments $\overline{x}$, pre- and postconditions $e_p$
  and $e_q$ and function body $s$, then
  $\coqAAAA{\vfSymExecFunc}{\overline{x}}{e_p}{s}{e_q}$ implies
  $\coqAAAA{\vfCbsemFunc}{\overline{x}}{e_p}{s}{e_q}$.
\end{theorem}
\begin{proof}
  The proof proceeds by examining the three cases provided by Lemma
  \ref{lem:symexec_stmt_sound} for symbolic execution of the function body $s$.
  The case for normal \emph{termination} of the function body can be easily
  dispelled, because, as Definition \ref{def:sym_exec_func} shows, the normal
  \emph{continuation} for $s$ in $\vfSymExecFunc$ is $\coqL{\_}{\coqFalse}$. The
  other two cases are trivial.
\end{proof}


\section{Proving correctness in Clight big step semantics}
\label{section:clight}

Theorem \ref{thm:symexec_func_sound} proves that symbolic execution of a
function allows us to conclude semantic correctness of that function with
regards to its pre- and postcondition within \cbsem. But can we be
confident that \cbsem\ itself is a sound semantics for function execution?

In this section we will prove soundness of \cbsem\ with regards to Clight for a
subset of \vfcx. This further narrowing of \vfcx\ is mainly restricted in terms
of some operations for boolean expressions (we drop $\vfxBinop{e_1}{==}{e_2}$,
$\vfxBinop{e_1}{\neq}{e_2}$, $\vfxBinop{e_1}{\ \&\&\ }{e_2}$,
$\vfxBinop{e_1}{\ ||\ }{e_2}$ and $\vfxUnop{!}{e}$) and block statements
$\vfxBlock{s}$. The resulting subset of \vfcx\ remains expressive enough for
meaningful functions with loops, conditionals and the capacity for displaying
undefined behavior, including the example from Section \ref{section:example}.

First we will draw relations between VeriFast stores
(\ref{subsection:store_tenv_equivalence}), expressions
(\ref{subsection:vf_cl_expression_equivalence}) and statements
(\ref{subsection:vf_cl_statement_equivalence}) and their counterparts in Clight.
Next, we will introduce two small transformations which are necessary to line up
the ASTs exported from VeriFast with those exported by \texttt{clightgen}
(\ref{subsection:statement_transformations}). Finally, because we support only a
limited subset of C, we define custom notions of program correctness in both
VeriFast and Clight and provide our second main result, proving that proving
correctness of program execution in \cbsem\ implies correctness of program
execution in Clight (\ref{subsection:vf_cl_programs}).

\subsection{Stores and temporary environments}
\label{subsection:store_tenv_equivalence}

In \vfcx, local variable values are stored in a store $\sigma$. But the local
variable $\vfxVar{x}$ introduced by $\vfxLet{x}{e}{s}$ exist only for the
duration of statement $s$. Once $s$ terminates, the variable goes out of scope,
meaning that the value for $\vfxVar{x}$ in the resulting store is $\coqNone$.

CompCert has two mechanisms for local variables, but we use the so-called
\emph{temporary variables}, which are stored in a \emph{temporary environment}
$\clTenv \in \coqT{temp_env}$. A major difference with local variables in \vfcx\
is that temporaries in CompCert exist from the very beginning of a function.
This means there is no ``unsetting'' of temporary variables, e.g. when exiting
some block.

Because of this important difference, we need to track a \emph{partial}
correspondence between a VeriFast store $\sigma$ and a Clight temporary
environment $\clTenv$ throughout an execution, limited to the variables which
are currently in scope \emph{in the VeriFast program}.

\begin{definition}
  \label{def:st_tenv_rel}
  The relation $\vfclCoveredBySymbol: \vfStore \to \clT{temp_env} \to \coqProp$
  is inductively defined as:
  \begin{mathparpagebreakable}
    \inferrule{\ }{
      \vfclCoveredBy{\vfEmptyStore}{\clTenv}
    }
    \and
    \inferrule{
      i_\CL = \coqA{\coqT{Int.repr}}{z} \\
      \vfclCoveredBy{\sigma}{\clTenv}
    }{
      \vfclCoveredBy{
        \vfStoreUpdate{x}{\coqA{\coqSome}{z}} \sigma
      }{
        (\coqAAA{\clT{PTree.set}}{\clId{x}}{(\coqA{\clT{Vint}}{i_\CL})}{\clTenv})
      }
    }
    \and
    \inferrule{
      \vfclCoveredBy{\sigma}{\clTenv}
    }{
      \vfclCoveredBy{\vfStoreUpdate{x}{\coqNone} \sigma}{\clTenv}
    }
  \end{mathparpagebreakable}
  In the above definition, CompCert's $\coqT{PTree.set}$ function updates a
  temporary variable identified by $\$ x$ (mapping string $x$ to some internal
  identifier representation) in $\clTenv$ with a CompCert $\clT{Int}$ value
  built from a value $z$.
\end{definition}

\subsection{Expressions}
\label{subsection:vf_cl_expression_equivalence}

For expressions, the correspondence between VeriFast and Clight for the
supported subset of \vfcx\ is described by two straightforward inductive
relations, one for integer evaluation and another one for boolean evaluation, in
line with the definitions for $\vfEvalZ$ and $\vfEvalBool$ from Definition
\ref{def:eval_Z_bool}.

\begin{definition}
  \label{def:expr_equiv_int}

  The relation $\vfclEquivIntSymbol: \vfExpr \to \coqT{Clight.expr} \to
  \coqProp$ is inductively defined as:
  \begin{mathparpagebreakable}
    \inferrule{
      i_\CL = \coqA{\coqT{Int.repr}}{[\vfxVar{z}]}
    }{
      \vfclEquivInt{\vfxVar{z}}{\coqAA{\clEconstInt}{i_\CL}{\clTint}}
    }
    \and
    \inferrule{\ }{
      \vfclEquivInt{\vfxVar{x}}{\coqAA{\clT{Etempvar}}{\$ [\vfxVar{x}]}{\clTint}}
    }
    \and
    \inferrule{
      \vfclEquivInt{e_1}{e_{\CL,1}} \\
      \vfclEquivInt{e_2}{e_{\CL,2}}
    }{
      \vfclEquivInt{(\vfxBinop{e_1}{+}{e_2})}{(\coqAAAA{\clEbinop}{\clT{Oadd}}{e_{\CL,1}}{e_{\CL,2}}{\clTint})}
    }
    \and
    \inferrule{
      \vfclEquivInt{e_1}{e_{\CL,1}} \\
      \vfclEquivInt{e_2}{e_{\CL,2}}
    }{
      \vfclEquivInt{(\vfxBinop{e_1}{-}{e_2})}{(\coqAAAA{\clEbinop}{\clT{Osub}}{e_{\CL,1}}{e_{\CL,2}}{\clTint})}
    }
    \and
    \inferrule{
      \vfclEquivInt{e_1}{e_{\CL,1}} \\
      \vfclEquivInt{e_2}{e_{\CL,2}}
    }{
      \vfclEquivInt{(\vfxBinop{e_1}{/}{e_2})}{(\coqAAAA{\clEbinop}{\clT{Odiv}}{e_{\CL,1}}{e_{\CL,2}}{\clTint})}
    }
  \end{mathparpagebreakable}
\end{definition}

\begin{lemma}[Soundness of integer evaluation]
  \label{lem:vf_cl_expr_to_int_sound}

  Assume a store $\sigma$ with $\coqA{\vfWellTyped}{\sigma}$ and a temporary
  environment $\clTenv$, such that $\vfclCoveredBy{\sigma}{\clTenv}$. Let $e$
  and $e_\CL$ be expressions such that $\vfclEquivInt{e}{e_\CL}$. Then
  $\coqAA{\vfEvalZ}{e}{\sigma} = \coqA{\coqSome}{z}$ implies that $e_\CL$
  evaluates to $\coqA{\clT{Vint}}{(\coqA{\clT{Int.repr}}{z})}$ by Clight's
  $\clT{eval_expr}$ semantics, starting from $\clTenv$.
\end{lemma}
\begin{proof}
  By induction on $\vfclEquivInt{e}{e_\CL}$. The main issues are the unpacking
  of CompCert's definitions for evaluation and integer values and proving that
  VeriFast's $\vfEvalZ$ implies CompCert's boundedness of the results.
\end{proof}

\begin{definition}
  \label{def:expr_equiv_b}

  The relation $\vfclEquivBoolSymbol: \vfExpr \to \coqT{Clight.expr} \to
  \coqProp$ is inductively defined as:
  \begin{mathparpagebreakable}
    \inferrule{\ }{
      \vfclEquivBool{\vfxTrue}{(\coqAA{\clEconstInt}{\clT{Int.one}}{\clTint})}
    }
    \and
    \inferrule{\ }{
      \vfclEquivBool{\vfxFalse}{(\coqAA{\clEconstInt}{\clT{Int.zero}}{\clTint})}
    }
    \and
    \inferrule{
      \vfclEquivInt{e_1}{e_{\CL,1}} \\
      \vfclEquivInt{e_2}{e_{\CL,2}}
    }{
      \vfclEquivBool{
        (\vfxBinop{e_1}{<}{e_2})
      }{
        (\coqAAAA{\clEbinop}{\clT{Olt}}{e_{\CL,1}}{e_{\CL,2}}{\clTint})
      }
    }
    \and
    \inferrule{
      \vfclEquivInt{e_1}{e_{\CL,1}} \\
      \vfclEquivInt{e_2}{e_{\CL,2}}
    }{
      \vfclEquivBool{
        (\vfxBinop{e_1}{\leq}{e_2})
      }{
        (\coqAAAA{\clEbinop}{\clT{Ole}}{e_{\CL,1}}{e_{\CL,2}}{\clTint})
      }
    }
  \end{mathparpagebreakable}
\end{definition}

\begin{lemma}[Soundness of boolean evaluation]
  \label{lem:vf_cl_expr_to_bool_sound}

  Assume a store $\sigma$ with $\coqA{\vfWellTyped}{\sigma}$ and a temporary
  environment $\clTenv$, such that $\vfclCoveredBy{\sigma}{\clTenv}$. Let $e$
  and $e_\CL$ be expressions such that $\vfclEquivBool{e}{e_\CL}$. Then
  $\coqAA{\vfEvalBool}{e}{\sigma} = \coqA{\coqSome}{b}$, with $b$ a Coq boolean,
  implies that there is some Clight value $b_\CL$, such that:
  \begin{enumerate}
    \item Clight's $\clT{eval_expr}$ evaluates expression $e_\CL$ to value $b_\CL$;
    \item and Clight's $\clT{bool_val}$ function evaluates value $b_\CL$ to
    $\coqA{\coqSome}{b}$.
  \end{enumerate}
  In other words, if \vfcx\ evaluates a VeriFast expression $e$ to a boolean
  value $b$, then Clight evaluates the corresponding expression $e_\CL$ to a
  corresponding value $b_\CL$.
\end{lemma}
\begin{proof}
  By induction on $\vfclEquivBool{e}{e_\CL}$. As with Lemma
  \ref{lem:vf_cl_expr_to_int_sound}, most of the proof revolves around unpacking
  CompCert's structures.
\end{proof}

\subsection{Statements}
\label{subsection:vf_cl_statement_equivalence}

As with expressions, the correspondence between \vfcx\ and Clight statements is
described by an inductively defined relation. We then describe two lemmas
describing soundness of statement execution, one for the terminating and one for
the diverging case.

\begin{definition}
  \label{def:stmt_equiv}

  The relation $\vfclEquivStmtSymbol: \vfStmt\to \coqT{Clight.statement} \to
  \coqProp$ is inductively defined as:
  \begin{mathparpagebreakable}
    \inferrule{\ }{
      \vfclEquivStmt{\vfxSkip}{\clT{Sskip}}
    }
    \and
    \inferrule{
      \vfclEquivStmt{s_1}{s_{\CL,1}} \\
      \vfclEquivStmt{s_2}{s_{\CL,2}}
    }{
      \vfclEquivStmt{(\vfxSeq{s_1}{s_2})}{(\coqAA{\clT{Ssequence}}{s_{\CL,1}}{s_{\CL,2}})}
    }
    \and
    \inferrule{
      \vfclEquivInt{e}{e_\CL} \\
      \vfclEquivStmt{s}{s_\CL}
    }{
      \vfclEquivStmt{
        (\vfxLet{x}{e}{s})
      }{
        (\coqAA{\clT{Ssequence}}{(\coqAA{\clT{Sset}}{\clId{x}}{e_\CL})}{s_\CL})
      }
    }
    \and
    \inferrule{
      \vfclEquivInt{e}{e_\CL}
    }{
      \vfclEquivStmt{
        (\vfxBinop{\vfxVar{x}}{=}{e})
      }{
        (\coqAA{\clT{Sset}}{\clId{[\vfxVar{x}]}}{e_\CL})
      }
    }
    \and
    \inferrule{
      \vfclEquivInt{e}{e_\CL}
    }{
      \vfclEquivStmt{
        (\vfxReturn{e})
      }{
        (\coqA{\clT{Sreturn}}{(\coqA{\coqSome}{e_\CL})})
      }
    }
    \and
    \inferrule{
      \vfclEquivBool{e}{e_\CL} \\
      \vfclEquivStmt{s_1}{s_{\CL,1}} \\
      \vfclEquivStmt{s_2}{s_{\CL,2}}
    }{
      \vfclEquivStmt{
        (\vfxIf{e}{s_1}{s_2})
      }{
        (\coqAAA{\clT{Sifthenelse}}{e}{s_{\CL,1}}{s_{\CL,2}})
      }
    }
    \and
    \inferrule{
      \vfclEquivBool{e}{e_\CL} \\
      \vfclEquivStmt{s}{s_{\CL}}
    }{
      \vfclEquivStmt{
        (\vfxWhile{e}{i}{s})
      }{
        (\coqAA{\clT{Swhile}}{e_\CL}{s_{\CL}})
      }
    }
  \end{mathparpagebreakable}
  Note that we use a combination of statements to model our own scoped variable
  declaration $\vfxLet{x}{e}{s}$.
\end{definition}

Next, we provide two lemmas proving the soundness of \vfcx's big step semantics
for \emph{statement execution}, $\vfCbsemExecStmtSymbol$ and
$\vfCbsemExecinfStmtSymbol$, to their counterparts in Clight's big step
semantics, $\clExecStmt$ and $\clExecinfStmt$.

For both the terminating and diverging case, we need to point out again the fact
that our bigstep semantics \cbsem\ does not deal with heap memories. This means
that the soundness lemmas for statement execution are universally quantified
over Clight memory states \emph{without further side conditions}; and in the
terminating case, the post-execution memory state is identical with the
pre-execution memory state.

In addition, \cbsem\ does not handle function calls or I/O, whereas Clight keeps
track of such observable events using a inductive type $\clT{trace}$ for
terminating executions and a coinductive type $\clT{traceinf}$ for diverging
executions. So terminating Clight executions in Lemma
\ref{lem:vf_cl_exec_stmt_sound} will always be derived with the empty Clight
$\clT{trace}$ $\clT{E0}$. Somewhat surprisingly perhaps, the lemma for soundness
of infinite statement execution (Lemma \ref{lem:vf_cl_execinf_stmt_sound})
universally quantifies \emph{over every possible Clight $\clT{traceinf}$}.

\begin{lemma}[Soundness of a terminating statement execution]
  \label{lem:vf_cl_exec_stmt_sound}

  Assume a store $\sigma$ with $\coqA{\vfWellTyped}{\sigma}$ and a temporary
  environment $\clTenv$, such that $\vfclCoveredBy{\sigma}{\clTenv}$. Let $s$
  and $s_\CL$ be statements such that $\vfclEquivStmt{s}{s_\CL}$. If
  $\vfCbsemExecStmt{\sigma}{s}{\sigma'}{o}$, then:
  \begin{enumerate}
    \item for $o = \vfOutN$, there exists a $\clTenv'$ such that Clight's
    $\clExecStmt$ relates execution of $s_\CL$, starting from $\clTenv$, with
    $\clTenv'$ and with Clight outcome $\clT{Out_normal}$.
    \item likewise, for $o = \coqA{\vfOutR}{z}$, there exists a $\clTenv'$ such
    that $\clExecStmt$ relates execution of $s_\CL$, starting from $\clTenv$,
    with $\clTenv'$ and outcome $\clT{Out_return}$, together with the proper
    Clight representation of $z$ as return value.
  \end{enumerate}
  In both cases, $\vfclCoveredBy{\sigma'}{\clTenv'}$ and the Clight execution
  is derived with the empty Clight event trace $\clT{E0}$.
\end{lemma}
\begin{proof}
  As in earlier proofs, the presence of the while loop requires us to proceed by
  induction on the derivation of $\vfCbsemExecStmt{\sigma}{s}{\sigma'}{o}$.
\end{proof}

\begin{lemma}[Soundness of a diverging statement execution]
  \label{lem:vf_cl_execinf_stmt_sound}

  Assume a store $\sigma$ with $\coqA{\vfWellTyped}{\sigma}$ and a temporary
  environment $\clTenv$, such that $\vfclCoveredBy{\sigma}{\clTenv}$. Let $s$
  and $s_\CL$ be statements such that $\vfclEquivStmt{s}{s_\CL}$. If
  $\vfCbsemExecinfStmt{\sigma}{s}$, then there is diverging execution using
  Clight's $\clExecinfStmt$, starting from $\clTenv$.
\end{lemma}
\begin{proof}
  The proof is coinductive and proceeds by case analysis on the statement $s$.
\end{proof}

\subsection{Lining up ASTs}
\label{subsection:statement_transformations}

We have now constructed a notion $\vfclEquivStmtSymbol$ of statement equivalence
and proven that, if $\vfclEquivStmt{s}{s_\CL}$, execution of $s$ in \cbsem\ is
sound wrt execution of $s_\CL$ in Clight. However, imagine some C code which
falls within the subset of supported expressions and statements, such as the
example from Section \ref{section:example}. If VeriFast parses this code and
produces a statement $s$ and if CompCert parses this code and produces a
statement $s_\CL$: will the relation $\vfclEquivStmt{s}{s_\CL}$ hold? For our
simple test cases, the answer is \emph{yes, except for the following two
differences}:
\begin{itemize}
  \item Our VeriFast export code currently exports a lot of
  $(\vfxSeq{s}{\vfxSkip})$ instances, which in Clight are all represented as
  \emph{just} the Clight equivalent of $s$.
  \item For the statement representing a \vfcinline{main} function body, the
  Clight compiler \texttt{clightgen} always wraps a main function with a
  \vfcinline{return 0} statement.
\end{itemize}
So before moving to program correctness, we introduce two AST transformation
functions which deal with these differences. For each transformation, we prove
that they preserve relevant execution properties within \cbsem.

\begin{definition}
  \label{def:simplify_vf_stmt}

  Fixpoint function $\vfSimplifyStmt: \vfStmt \to \vfStmt$ recursively removes
  instances of $(\vfxSeq{\_}{\vfxSkip})$ from a statement $s$:
  \begin{align*}
    \vfSimplifyStmt\ (\vfxSeq{s}{\vfxSkip}) & := s. \\
    \vfSimplifyStmt\ (\vfxSeq{s_1}{s_2}) & := \vfxSeq{(\coqA{\vfSimplifyStmt}{s}_1)}{(\coqA{\vfSimplifyStmt}{s}_2)}. \\
    \vfSimplifyStmt\ (\vfxLet{x}{e}{s}) & := \vfxLet{x}{e}{(\coqA{\vfSimplifyStmt}{s})}. \\
    \vfSimplifyStmt\ (\vfxIf{e}{s_1}{s_2}) & := \vfxIf{e}{(\coqA{\vfSimplifyStmt}{s}_1)}{(\coqA{\vfSimplifyStmt}{s}_2)}. \\
    \vfSimplifyStmt\ (\vfxWhile{e}{i}{s}) & := \vfxWhile{e}{i}{(\coqA{\vfSimplifyStmt}{s})}. \\
    \vfSimplifyStmt\ \vfxBlock{s} & := \vfxBlock{\coqA{\vfSimplifyStmt}{s}}. \\
    \coqA{\vfSimplifyStmt}{s} & := s \text{, for all other cases.}
  \end{align*}
\end{definition}

\begin{definition}
  \label{def:simplify_vf_stmt_rel}

  $\vfSimplifyStmtRel: \vfStmt \to \vfStmt \to \coqProp$ is an inductively
  defined propositional relation which reflects the computational function
  $\vfSimplifyStmt$. The only purpose of this extra construct is that it makes
  the proofs for preservation of termination and divergence (see Lemmas
  \ref{lem:simplify_vf_stmt__preserves__termination} and
  \ref{lem:simplify_vf_stmt__preserves__divergence}) easier. We omit the
  definition for this relation.
\end{definition}

\begin{lemma}
  \label{lem:simplify_vf_stmt_rel_intro}

  For each statement $s$, $s \vfSimplifyStmtRel (\coqA{\vfSimplifyStmt}{s})$.
\end{lemma}
\begin{proof}
  By straightforward induction on $s$.
\end{proof}

\begin{lemma}[$\vfSimplifyStmtRel$ preserves termination]
  \label{lem:simplify_vf_stmt__preserves__termination}

  Given statements $s_1$ and $s_2$ such that $s_1 \vfSimplifyStmtRel s_2$, then
  $\vfCbsemExecStmt{\sigma}{s_1}{\sigma'}{o}$ implies
  $\vfCbsemExecStmt{\sigma}{s_2}{\sigma'}{o}$.
\end{lemma}
\begin{proof}
  By induction on $\vfCbsemExecStmt{\sigma}{s_1}{\sigma'}{o}$.
\end{proof}

\begin{lemma}[$\vfSimplifyStmtRel$ preserves divergence]
  \label{lem:simplify_vf_stmt__preserves__divergence}

  Given statements $s_1$ and $s_2$ such that $s_1 \vfSimplifyStmtRel s_2$, then
  $\vfCbsemExecinfStmt{\sigma}{s_1}$ implies $\vfCbsemExecinfStmt{\sigma}{s_2}$.
\end{lemma}
\begin{proof}
  Perhaps surprisingly, this proof was slightly more difficult than expected.
  Starting the proof with coinduction fails for the case in which $s_1 =
  (\vfxSeq{s_1'}{\vfxSkip})$, due to the difficulty in fulfilling the
  guardedness condition. So we start by induction on $s_1 \vfSimplifyStmtRel
  s_2$, which is a finite derivation, allowing the application of the proper
  coinductive constructors to produce $\vfCbsemExecinfStmt{\sigma}{s_2}$. We
  then only need to make a co-recursive call within the case for the while loop,
  which can be properly guarded.
\end{proof}

\begin{definition}
  \label{def:programify_vf_stmt}

  Fixpoint function $\vfProgramifyStmt: \vfStmt \to \vfStmt$ appends a return
  statement to some $s$:
  \begin{align*}
    \vfProgramifyStmt s & := \vfxSeq{s}{\vfxReturn{0}}.
  \end{align*}
  This transformation will be necessary because CompCert does it automatically
  to each main function.
\end{definition}

\begin{lemma}[$\vfProgramifyStmt$ preserves return outcomes]
  \label{lem:programify_vf_stmt__preserves__return}

  Given $\vfCbsemExecStmt{\sigma}{s}{\sigma'}{\coqA{\vfOutR}{z}}$ then
  $\vfCbsemExecStmt{\sigma}{(\vfProgramifyStmt s)}{\sigma'}{\coqA{\vfOutR}{z}}$;
  that is, $\vfProgramifyStmt$ does not change the outcome of an execution which
  returns a value.
\end{lemma}
\begin{proof}
  By case analysis on the derivation of
  $\vfCbsemExecStmt{\sigma}{s}{\sigma'}{\coqA{\vfOutR}{z}}$.
\end{proof}

\begin{lemma}[$\vfProgramifyStmt$ preserves divergence]
  \label{lem:programify_vf_stmt__preserves__divergence}

  Given that $\vfCbsemExecinfStmt{\sigma}{s}$ holds, then
  $\vfCbsemExecinfStmt{\sigma}{(\vfProgramifyStmt{s})}$.
\end{lemma}
\begin{proof}
  By application of the $\vfCbsemExecinfStmtSymbol$ constructor for divergence
  of the first statement of a sequence.
\end{proof}

\subsection{Correctness of programs}
\label{subsection:vf_cl_programs}

We want to build confidence in our \cbsem\ semantics by proving soundness with
respect to Clight. To this end, Lemmas \ref{lem:vf_cl_exec_stmt_sound} and
\ref{lem:vf_cl_execinf_stmt_sound} already provide good results. But ultimately,
our primary goal remains to work towards \emph{full end-to-end soundness on the
level of programs}, giving us confidence that a program which is proven correct
by symbolic execution in VeriFast, when afterwards compiled with CompCert and
executed, will not crash.

The subset of C that we chose for this proof-of-concept does not feature
\emph{function calls}, meaning that any statements about program correctness
must consider only the \vfcinline{main} function in a program. With this
shortcut in mind, we still need to provide precise notions for the following
things:
\begin{itemize}
  \item program correspondence between a \cbsem\ program and a Clight program
  (see Definition \ref{def:prog_equiv});
  \item program correctness in \cbsem\footnote{The notion of \cbsem\ program
  correctness introduced here will be tailored for usage with Clight, which is
  why we define it here and not in Section \ref{section:cbsem}.} (see Definition
  \ref{def:cbsem_exec_prog_correct}), together with a connection to \cbsem's
  function correctness $\vfCbsemFunc$, which is required to make the connection
  to symbolic execution;
  \item program correctness for Clight (Definition
  \ref{def:compcert_exec_prog_correct}).
\end{itemize}
With all these definitions in place we will then state and prove the second main
soundness theorem of this work (Theorem \ref{thm:vf_cl_sound}).

\begin{definition}
  \label{def:prog_equiv}

  The relation $\vfclEquivProgSymbol: \vfStmt \to \clT{Clight.program} \to
  \coqProp$ is inductively defined with a single constructor, relating a \vfcx\
  statement (representing the body of the \vfcinline{main} function as parsed
  with VeriFast) with an entire Clight program $p_\CL$ with main function
  $\clT{cl_main}$:
  \begin{mathparpagebreakable}
    \inferrule{
      (\dots) \\
      \coqA{\clT{Clight.fn_body}}{\clT{cl_main}} = s_\CL \\
      \vfclEquivStmt{s}{s_\CL}
    }{
      \vfclEquivProg{s}{p_\CL}
    }
  \end{mathparpagebreakable}
  Although not shown explicitly in the above inference rule, this constructor
  also requires evidence about the initial execution state and the memory layout
  of the CompCert program. As our C subset does not support programs using heap
  memory and main functions taking arguments, this can be trivially provided by
  the VeriFast-exported Coq script, after it has imported the Coq script
  generated by CompCert.
\end{definition}

\begin{definition}
  \label{def:cbsem_exec_prog_correct}

  A \vfcx\ program consists of a \vfcinline{main} function with body $s$ which
  takes no arguments. This program is \emph{correct}, written
  $\coqA{\vfCbsemProg}{s}$, if either:
  \begin{itemize}
    \item there exist $\sigma'$ and $z$ such that
    $\vfCbsemExecStmt{\emptyset}{s}{\sigma'}{\coqA{\vfOutR}{z}}$;
    \item or $\vfCbsemExecinfStmt{\sigma}{s}$.
  \end{itemize}
\end{definition}

Predicate $\vfCbsemProg$ will be proven sound with regards to Clight. But in
order to apply the full end-to-end chain of soundness proofs, we need to link it
to symbolic execution. In order to use Theorem \ref{thm:symexec_func_sound},
which links the correctness of symbolic execution $\vfSymExecFunc$ to semantic
correctness $\vfCbsemFunc$, we need another small lemma which links
$\vfCbsemFunc$ to $\vfCbsemProg$.

\begin{lemma}
  \label{lem:cbsem_func_sound}

  Given $\vfCbsemFunc\ \coqT{[}\ \coqT{]}\ \vfxTrue\ s\ q$, that is the correct
  execution in \cbsem\ of a function without arguments and with body $s$,
  precondition $\vfxTrue$ and postcondition $q$, then we may conclude
  $\coqA{\vfCbsemProg}{(\coqA{\vfProgramifyStmt}{(\coqA{\vfSimplifyStmt}{s})})}$.
\end{lemma}
\begin{proof}
  Hypothesis $\vfCbsemFunc\ \coqT{[}\ \coqT{]}\ \vfxTrue\ s\ q$ is a disjunction
  which allows us to consider two cases:
  \begin{enumerate}
    \item For the case of termination of $s$ with a return outcome, we proceed
    by Lemmas \ref{lem:simplify_vf_stmt__preserves__termination},
    \ref{lem:simplify_vf_stmt__preserves__divergence} and
    \ref{lem:simplify_vf_stmt_rel_intro}.
    \item For the case of divergence of $s$, we proceed by Lemmas
    \ref{lem:programify_vf_stmt__preserves__return} and
    \ref{lem:programify_vf_stmt__preserves__divergence}.
  \end{enumerate}
\end{proof}

\begin{definition}
  \label{def:compcert_exec_prog_correct}

  Predicate $\clExecProg: \clT{Clight.program} \to \coqProp$ expresses
  correctness of a program $p$ with regards to Clight's own big step semantics
  ($\clT{ClightBigstep}$ \cite{blazy-clight-jar-2009}) if:
  \begin{itemize}
    \item either there exists a return value $v$ such that, for the empty trace
    $\clT{E0}$,
    $\coqAAA{\clT{bigstep_program_terminates}}{p}{\clT{E0}}{(\coqA{\clT{Int.repr}}{v})}$;
    \item or $\coqAA{\clT{bigstep_program_diverges}}{p}{\clT{trc}_\infty}$ holds
    for all infinite event traces $\clT{trc}_\infty$.
  \end{itemize}
  $\clT{bigstep_program_terminates}$ and $\clT{bigstep_program_diverges}$
  describe CompCert's notion of program correctness for terminating and
  diverging programs.

  The universal quantifier in the diverging case results from the fact that
  $\clT{traceinf}$ is defined coinductively akin to a stream, so there is no
  empty $\clT{traceinf}$ constructor. Since there is no concept in \cbsem\ which
  is equivalent to Clight execution traces, any infinite event trace is
  accepted.
\end{definition}

\begin{theorem}[Soundness of program execution]
  \label{thm:vf_cl_sound}

  Consider a main function with body $s$. Let $s'$ be shorthand notation for
  $\coqA{\vfProgramifyStmt}{(\coqA{\vfSimplifyStmt}{s})}$. If $s'
  \vfclEquivProgSymbol p$ for some Clight program $p$, then program correctness
  $\coqA{\vfCbsemProg}{s'}$ in VeriFast implies program correctness
  $\coqA{\clExecProg}{p}$ in Clight.
\end{theorem}
\begin{proof}
  Hypothesis $\coqA{\vfCbsemProg}{s'}$ provides us with two cases, both of which
   provide all the evidence needed to conclude
   $\coqA{\clT{bigstep_program_terminates}}{p}$ for the returning case or
   $\coqA{\clT{bigstep_program_diverges}}{p}$ for the diverging case.
\end{proof}


\section{Related work}
\label{section:relatedwork}

Due to the limited scale of this TR, it is early to provide a deep comparison of
advantages and disadvantages with respect to approaches taken by other projects.
We do however want to list the following related projects.

\subsection*{Certificate-generating program verification tools}

Boogie verifies a program, written in its intermediate language, by transforming
it through a number of steps into a Verification Condition (VC). This VC can
then be validated by an SMT solver. Similar to our approach, a recent extension
to the Boogie verifier \cite{parthasarathy-boogie-cav-2021} allows the automatic
generation of machine checkable certificates for a verification run, thereby
avoiding direct verification of Boogie itself. These certificates import the
input- and output programs for three of the most complicated steps in Boogie's
verification pipeline into Isabelle. For each step, it is proven that
correctness of the output program (or VC) implies correctness of the input
program. The certificates however do not currently prove the VC itself and
therefore do not constitute a full machine checkable correctness proof of the
verified program.

\subsection*{Certified verification of C programs}

Verified Software Toolchain (VST) \cite{appel-sf5-2021} allows semi-automatic
production of a correctness proof for C programs, stated in a program logic
called Verifiable C. It is developed entirely within Coq and tightly integrated
with CompCert, making it possible to conclude that the correctness properties of
a verified program are preserved during compilation with CompCert. Verifiable C
contains a very complete set of C language features and allows reasoning about
memory states through separation logic. The actual proving in VST, while
assisted by a library of powerful tactics (VST-Floyd), still requires manual
guidance. Our own aim however is for our certificates to be standalone Coq
scripts, meaning that we want to avoid interactive proof guidance at the level
of Coq. We plan to achieve this by translating into Coq the pre-existing
VeriFast constructs such as loop invariants, lemmas and heap predicates; and by
replicating the reasoning performed by VeriFast's SMT solver.

RefinedC \cite{sammler-refinedc-pldi-2021} takes an approach to foundational
verification not unlike that taken by VST, in the sense that it allows stating
the correctness of C programs featuring heap memory and concurrency and that
proving these properties takes place entirely within Coq. Its language semantics
also describe a language with a low-level memory model similar to (but not
identical with) that of CompCert \cite{leroy-memory-tr-2012}. However, RefinedC
differs significantly from VST in that it promises fully automatic proving,
meaning that somebody who is certifying a C program with RefinedC should not
have to delve into Coq. To this end, RefinedC annotations describing pre- and
postconditions and invariants leverage a sophisticated typing system, which is
translated into a subset of the Iris framework for separation logic
\cite{jung-iris-jfp-2018} that facilitates efficient, syntax-directed proof
search. Automation therefore may require an expert extending the typing system
for new use cases. Likewise, by relying on the translation of existing VeriFast
constructs and by inspection of the SMT solver, we aim to offer a similar level
of automation. We hope that, to some degree, our approach will make it possible
to handle new use cases \emph{merely} through these existing VeriFast
constructs, rather than by having an expert make changes to the Coq library or
to the exporting code.


\section{Future work}
\label{section:futurework}

The most obvious need for expansion of course exists in the subset of C that we
support. Our ultimate goal is to support the full set of C supported by
VeriFast. This will require symbolic execution in Coq and \cbsem\ to support
function calls, pointer and \vfcinline{struct} types, heap memory and separation
logic. We will also need to translate supporting VeriFast constructs such as
heap predicates, lemmas and fixpoints into Coq and make these translated
versions useable for our automatically generated SEP proofs. Keeping these
proofs automated for the general case will require us to move beyond the
simplistic syntax-driven tactics presented in Section
\ref{subsection:proving_sep}.

In terms of setting up a \emph{soundness chain} with respect to a third party
operational semantics for C, we aim to prove \cbsem\ sound with respect to CH2O
\cite{krebbers-ch2o-phd-2015}, which provides a very strict interpretation of
the C11 standard \cite{iso-c-standard-2012}. If VeriFast can prove a program
correct in CH2O, this should mean that the program will be correct with respect
to any common compiler which respects the C standard -- which, while in no way a
\emph{formal} statement, is a very useful result.

We can then still further extend the soundness chain by proving CH2O operational
semantics sound with respect to one of CompCert's operational semantics. This
would allow us to keep the benefit of being able to make a formal statement
about program correctness down to the level of assembly code, a benefit that we
currently possess by virtue of Theorem \ref{thm:vf_cl_sound}, proving the
soundness of \cbsem\ with respect to CompCert.


\bibliographystyle{acm}
\bibliography{index}


\appendix


\section{Building and using the extension}
\label{appendix:building_and_using}

The extension described in this TR is not yet included in the VeriFast trunk.
But as mentioned in the introduction, it can be downloaded here:

\begin{center}
  \vfurl
\end{center}

In this appendix we provide instructions on how to build this extension to
VeriFast and how to use it for verification of simple programs which fall within
the limited subset supported by \vfcx.

\subsection{Building the extension}

First we need to build VeriFast itself. First, download the code from the above
address and check out the branch which contains the extension:
\begin{verbatim}
$ git clone verifast-coq-export-tr-2021.bundle verifast
$ cd verifast
$ git checkout coq-export-tr-2021
\end{verbatim}
Next, we follow the standard build instructions for VeriFast. For macOS, we can
build VeriFast like this:
\begin{verbatim}
$ ./setup-build.sh
$ cd src
$ export PATH=/usr/local/vfdeps-509f16f/bin:$PATH
$ export DYLD_LIBRARY_PATH=/usr/local/vfdeps-509f16f/lib:
    $DYLD_LIBRARY_PATH
$ make
\end{verbatim}
For other platforms such as Windows and Linux, consult the VeriFast
documentation itself (these can be found through the main \coderef{README.md}
included with the VeriFast project).

Next we need to install CompCert and compile our own Coq code. We have developed
this project using Coq 8.13.2 and CompCert 3.9. Our current suggestion is to do
all of this using OCaml's package manager Opam. Opam can be easily installed
using your system's package manager, such as Homebrew on macOS. Opam allows us
to build a local switch (which is a local, self-contained installation of OCaml
and Coq).

Once you have installed Opam, execute the following commands from the main
VeriFast directory to install the Coq and CompCert requirements:
\begin{verbatim}
$ cd ..
$ opam switch create . 4.12.0
$ eval $(opam env)
$ opam install coq.8.13.2
$ opam repo add coq-released https://coq.inria.fr/opam/released
$ opam install coq-compcert.3.9
\end{verbatim}
Now we can compile the VeriFast Coq code itself:
\begin{verbatim}
$ coq_makefile -f "_CoqProject" -o "Makefile.coq"
$ make -f Makefile.coq
\end{verbatim}

\subsection{Using the extension}

A number of test files are included in folder \coderef{tests/coq}. These have not
yet been added to the main test suite, because our setup will change in the
future\footnote{Specifically, we may decide to no longer automatically load the
CompCert-generated Coq script from within our own artefact. More fundamentally,
we will likely opt to replace CompCert with another third-party C semantics.}.
To run one of these included examples (such as \coderef{test_countdown.c}, which
is the example from Section \ref{section:example}), use the
\texttt{-emit\_coq\_proof} option:
\begin{verbatim}
$ bin/verifast -shared -emit_coq_proof -bindir bin
    tests/coq/test_countdown.c
\end{verbatim}
Upon successful verification, this will generate a Coq artefact
\coderef{test_vf.v} in your present working directory. This artefact will need
to be successfully typechecked by Coq to prove the correctness of VeriFast's
verification. But because this artefact imports another Coq file, generated by
CompCert's \texttt{clightgen}, we first need to generate this second script and
compile it using Coq:
\begin{verbatim}
$ clightgen tests/coq/test_countdown.c -o test_cc.v
$ coqc test_cc.v
\end{verbatim}
With the CompCert export compiled, we can now finally typecheck the artefact
exported by VeriFast:
\begin{verbatim}
$ coqc test_vf.v -Q src/coq verifast
\end{verbatim}
Note that some examples from \coderef{tests/coq} require the
\texttt{-allow\_dead\_code} option to be used with VeriFast, since they
explicitly test how our export code deals with unreachable code.

Finally, a simple shell script \coderef{test_coq.sh} has been included to
perform the above 4 steps automatically:
\begin{verbatim}
$ ./test_coq.sh tests/coq/test_countdown.c
\end{verbatim}
This script leaves the resulting scripts around for further inspection.


\section{Overview of the extension code}
\label{appendix:overview}

In the course of our work, we developed a Coq library which must be loaded by
each proof artefact. In this appendix, we briefly discuss the Coq modules which
comprise this library. We also provide lookup tables to locate the
implementations of the definitions, lemmas and theorems in this report. The
files for this library can be found under \coderef{src/coq/*.v}. Finally, we
also briefly describe how we instrumented the VeriFast OCaml code to export
the certificates themselves.

\subsection*{\coderef{base.v}}

\coderef{cx.v} contains basic Coq notation settings shared by our entire
library.

\subsection*{\coderef{cx.v}}

\coderef{cx.v} contains the main type definitions for embedding \vfcx\ ASTs in
Coq. Its contents are mainly introduced in Section~\ref{section:vfcx}.

\begin{center}
  \begin{tabularx}{\textwidth}{p{54mm}p{34mm}X}
    \textbf{Name in Coq module} & \textbf{Concise} & \textbf{Introduced} \\
    \hline

    \coqinline{expr} & $e \in \vfExpr$ & Definition~\ref{def:expr} \\
    \coqinline{stmt} & $s \in \vfStmt$ & Definition~\ref{def:stmt} \\
    \coqinline{stmt_to_free_targets} & $\coqA{\vfFreeTargets}{s}$ & Definition~\ref{def:stmt_to_free_possible_targets}
  \end{tabularx}
\end{center}

\subsection*{\coderef{shared.v}}

Module \coderef{shared.v} currently collects a variety of basic data structures
and lemmas related to stores, store equivalence, variable binding and generating
propositions that quantify over variables.

\begin{center}
  \begin{tabularx}{\textwidth}{p{54mm}p{34mm}X}
    \textbf{Name in Coq module} & \textbf{Concise} & \textbf{Introduced} \\
    \hline

    \coqinline{store} & $\sigma \in \vfStore$ & Definition~\ref{def:store} \\
    \coqinline{empty_store} & $\vfEmptyStore$ & Definition~\ref{def:store} \\
    \coqinline{store_update} & $\vfStoreUpdate{x}{v} \sigma$ & Definition~\ref{def:store} \\
    \coqinline{cont} & $C \in \vfCont$ & Definition~\ref{def:cont} \\
    \coqinline{min_signed} & $\vfMinSigned$ & Definition~\ref{def:for_Z} \\
    \coqinline{max_signed} & $\vfMaxSigned$ & Definition~\ref{def:for_Z} \\
    \coqinline{is_int} & $\coqA{\coqT{is_int}}{z}$ & Definition~\ref{def:for_Z} \\
    \coqinline{for_Z} & $\coqAAA{\vfForZ}{x}{C}{\sigma}$ & Definition~\ref{def:for_Z} \\
    \coqinline{for_Zs} & $\coqAAA{\vfForZs}{\overline{x}}{C}{\sigma}$ & Definition~\ref{def:for_Z} \\
    \coqinline{havoc_Zs} & $\coqAAA{\vfHavocZs}{\overline{x}}{C}{\sigma}$ & Definition~\ref{def:for_Z} \\
    \coqinline{well_bound} & $\coqA{\vfWellTyped}{\sigma}$ & Definition~\ref{def:well_typed} \\
    \coqinline{store_eq_mod_xs} & $\sigma \equiv_{\overline{x}} \sigma'$ & Definition~\ref{def:store_eq_mod_xs} \\
    \coqinline{binds} & $\overline{x} \subseteq \coqA{\vfDom}{\sigma}$ & Definition~\ref{def:binds}
  \end{tabularx}
\end{center}

\subsection*{\coderef{translate_expr.v}}

This module implements a number of functions to translate \vfcx\ expressions
into Coq data types such as $\coqProp$. Translation differs from evaluation in
that no side conditions (such as checks for division by zero) are generated.

\begin{center}
  \begin{tabularx}{\textwidth}{p{54mm}p{34mm}X}
    \textbf{Name in Coq module} & \textbf{Concise} & \textbf{Introduced} \\
    \hline

    \coqinline{translate_expr_to_Prop_cps} & $\coqAAA{\vfTranslateProp}{e}{C}{\sigma}$ & Definition~\ref{def:translate_expr_to_Prop_cps}
  \end{tabularx}
\end{center}

\subsection*{\coderef{eval_expr.v}}

Module \coderef{eval_expr.v} implements functions for the evaluation of \vfcx\
expressions into Coq data types such as $\coqProp$. Side conditions (such as
checks for division by zero) are generated whenever needed. (However, CPS
versions of these functions are found in \coderef{symexec.v} because they are
uniquely related to symbolic execution.)

\begin{center}
  \begin{tabularx}{\textwidth}{p{54mm}p{34mm}X}
    \textbf{Name in Coq module} & \textbf{Concise} & \textbf{Introduced} \\
    \hline

    \coqinline{eval_expr_to_Z} & $\coqAAA{\vfEvalZ}{e}{C}{\sigma}$ & Definition~\ref{def:eval_Z_bool} \\
    \coqinline{eval_expr_to_bool} & $\coqAAA{\vfEvalBool}{e}{C}{\sigma}$ & Definition~\ref{def:eval_Z_bool}
  \end{tabularx}
\end{center}

\subsection*{\coderef{symexec.v}}

Module \coderef{symexec.v} implements everything related to symbolic execution
in Coq: (CPS-based) evaluation of expressions and functions constructing the SEP
for execution of a statement and a function. It also provides a number of useful
lemmas for working with SEPs in proofs.

\begin{center}
  \begin{tabularx}{\textwidth}{p{54mm}p{34mm}X}
    \textbf{Name in Coq module} & \textbf{Concise} & \textbf{Introduced} \\
    \hline

    \coqinline{eval_expr_to_Z_cps} & $\coqAAA{\vfEvalZCps}{e}{C}{\sigma}$ & Definition~\ref{def:eval_Z_cps} \\
    \coqinline{eval_expr_Prop_cps} & $\coqAAA{\vfEvalPropCps}{e}{C}{\sigma}$ & Definition~\ref{def:eval_Prop_cps} \\
    \coqinline{produce} & $\coqAAA{\vfProduce}{e}{C}{\sigma}$ & Definition~\ref{def:produce_consume} \\
    \coqinline{consume} & $\coqAAA{\vfConsume}{e}{C}{\sigma}$ & Definition~\ref{def:produce_consume} \\
    \coqinline{leak_check} & $\coqA{\vfLeakCheck}{\sigma}$ & Definition~\ref{def:leak_check} \\
    \coqinline{exec_stmt} & $\coqAAAA{\vfSymExecStmt}{s}{C_N}{C_R}{\sigma}$ & Definition~\ref{def:sym_exec_stmt} \\
    \coqinline{ret} & $\coqAAAA{\vfRet}{\sigma_0}{C}{z}{\sigma}$ & Definition~\ref{def:ret} \\
    \coqinline{exec_func_correct} & $\coqAAAA{\vfSymExecFunc}{\overline{x}}{e_p}{s}{e_q}$ & Definition~\ref{def:sym_exec_func}
  \end{tabularx}
\end{center}

\subsection*{\coderef{cbsem.v}}

Module \coderef{cbsem.v} implements the \cbsem\ operational semantics itself,
together with some useful lemmas.

\begin{center}
  \begin{tabularx}{\textwidth}{p{54mm}p{34mm}X}
    \textbf{Name in Coq module} & \textbf{Concise} & \textbf{Introduced} \\
    \hline

    \coqinline{outcome} & $o \in \vfOut$ & Definition~\ref{def:outcome} \\
    \coqinline{exec_stmt} & $\vfCbsemExecStmt{\sigma}{s}{\sigma'}{o}$ & Definition~\ref{def:cbsem_exec_stmt} \\
    \coqinline{execinf_stmt} & $\vfCbsemExecinfStmt{\sigma}{s}$ & Definition~\ref{def:cbsem_execinf_stmt} \\
    \coqinline{exec_stmt_derivable} & $\vfCbsemExecStmtDerivable{\sigma}{s}$ & Definition~\ref{def:exec_stmt_not_derivable} \\
    \coqinline{exec_func_correct} & $\coqAAAA{\vfCbsemFunc}{\overline{x}}{e_p}{s}{e_q}$ & Definition~\ref{def:cbsem_exec_func_correct}
  \end{tabularx}
\end{center}

\subsection*{\coderef{symexec_cbsem.v}}

Module \coderef{symexec_cbsem.v} implements the first soundness result which
proves the soundness of symbolic execution with respect to \cbsem.

\begin{center}
  \begin{tabularx}{\textwidth}{p{54mm}p{34mm}X}
    \textbf{Name in Coq module} & \textbf{Concise} & \textbf{Introduced} \\
    \hline

    \coqinline{symexec_stmt_sound_termination} & & Lemma~\ref{lem:symexec_stmt_sound_termination} \\
    \coqinline{symexec_stmt_sound_divergence} & & Lemma~\ref{lem:symexec_stmt_sound_divergence} \\
    \coqinline{symexec_stmt_sound} & & Lemma~\ref{lem:symexec_stmt_sound} \\
    \coqinline{symexec_func_sound} & & Theorem~\ref{thm:symexec_func_sound}
  \end{tabularx}
\end{center}

\subsection*{\coderef{cbsem_clight.v}}

Module \coderef{cbsem_clight.v} provides the structures needed to express the
correspondence between a \vfcx\ program and a Clight program. Using these
structures, it then provides a number of lemmas, ending with the second
soundness theorem which proves soundness of \cbsem\ with respect to CompCert's
Clight for a subset of \vfcx.

\begin{center}
  \begin{tabularx}{\textwidth}{p{54mm}p{34mm}X}
    \textbf{Name in Coq module} & \textbf{Concise} & \textbf{Introduced} \\
    \hline

    \coqinline{st_tenv_rel} & $\vfclCoveredBy{\sigma}{\clTenv}$ & Definition~\ref{def:st_tenv_rel} \\
    \coqinline{expr_equiv_int} & $\vfclEquivInt{e}{e_\CL}$ & Definition~\ref{def:expr_equiv_int} \\
    \coqinline{vf_cl_expr_to_int_sound} & & Lemma~\ref{lem:vf_cl_expr_to_int_sound} \\
    \coqinline{expr_equiv_b} & $\vfclEquivBool{e}{e_\CL}$ & Definition~\ref{def:expr_equiv_b} \\
    \coqinline{vf_cl_expr_to_bool_sound} & & Lemma~\ref{lem:vf_cl_expr_to_bool_sound} \\
    \coqinline{stmt_equiv} & $\vfclEquivStmt{s}{s_\CL}$ & Definition~\ref{def:stmt_equiv} \\
    \coqinline{vf_cl_exec_stmt_sound} & & Lemma~\ref{lem:vf_cl_exec_stmt_sound} \\
    \coqinline{vf_cl_execinf_stmt_sound} & & Lemma~\ref{lem:vf_cl_execinf_stmt_sound} \\
    \coqinline{simplify_vf_stmt} & $\coqA{\vfSimplifyStmt}{s}$ & Definition~\ref{def:simplify_vf_stmt} \\
    \coqinline{simplify_vf_stmt_rel} & $s \vfSimplifyStmtRel s'$ & Definition~\ref{def:simplify_vf_stmt_rel} \\
    \coqinline{simplify_vf_stmt_rel_intro} & & Lemma~\ref{lem:simplify_vf_stmt_rel_intro} \\
    \coqinline{simplify_vf_stmt__preserves__termination} & & Lemma~\ref{lem:simplify_vf_stmt__preserves__termination} \\
    \coqinline{simplify_vf_stmt__preserves__divergence} & & Lemma~\ref{lem:simplify_vf_stmt__preserves__divergence} \\
    \coqinline{programify_vf_stmt} & $\coqA{\vfProgramifyStmt}{s}$ & Definition~\ref{def:programify_vf_stmt} \\
    \coqinline{programify_vf_stmt__preserves__return} & & Lemma~\ref{lem:programify_vf_stmt__preserves__return} \\
    \coqinline{programify_vf_stmt__preserves__divergence} & & Lemma~\ref{lem:programify_vf_stmt__preserves__divergence} \\
    \coqinline{prog_equiv} & $\vfclEquivProg{s}{p_\CL}$ & Definition~\ref{def:prog_equiv} \\
    \coqinline{cbsem_exec_prog_correct} & $\coqA{\vfCbsemProg}{s}$ & Definition~\ref{def:cbsem_exec_prog_correct} \\
    \coqinline{cbsem_func_sound} & & Lemma~\ref{lem:cbsem_func_sound} \\
    \coqinline{compcert_exec_prog_correct} & $\coqA{\clExecProg}{p_\CL}$ & Definition~\ref{def:compcert_exec_prog_correct} \\
    \coqinline{vf_cl_sound} & & Theorem~\ref{thm:vf_cl_sound}
  \end{tabularx}
\end{center}

\subsection*{Instrumenting VeriFast to generate certificates}

The main export code can be found in the OCaml module \coderef{src/coq.ml}. Our
exporting code depends on constructing a \texttt{recorder} object when setting
up a verification run (see function \texttt{verify\_program} in
\coderef{src/verifast.ml}). For the duration of the actual verification, this
recording object is available everywhere in the code.

We then added recorder calls at various strategic points in the VeriFast code
base. These calls correspond to high level events such as the SMT solver making
assumptions and checking assertions, fresh symbol picking, pushing and popping
SMT solver contexts and branching.

This intermediate recording tree is folded to produce a tree with nodes
corresponding to the four basic types of proof steps discussed in Subsection
\ref{subsection:proving_sep}: \emph{assumptions} and \emph{assertions} (for
calls to the SMT solver); \emph{conjunctions} (for branching the execution) and
a node expressing that an execution branch was terminated due to the presence of
a \emph{contradictory path condition}. This final tree is then used, together
with the program's AST, to export the certificate (see function \texttt{verify}
in \coderef{src/vfconsole.ml}).


\section{Full listing for \texttt{test\_vf.v}}
\label{appendix:listings}

\coqlistinglabel{sources/test_vf.v}{The full listing for \coderef{test_vf.v}, the script exported by VeriFast for the example in Section \ref{section:example}.}{lst:test_vf.v}

\end{document}